%% file: main.tex
\documentclass[11pt, dvipsnames]{scrartcl}
\input{packages.tex}

\input{macros.tex}

\begin{document}
%\pagestyle{fancy}
%\fancyhead[R]{\textcolor{gray}{NOT FOR DISTRIBUTION}}
%\fancyfoot[C]{\thepage}
%\renewcommand{\headrulewidth}{0pt}
	
\newcommand{\PG}[1]{\textcolor{red}{PG: {#1}}}
	
\title{Coresets for Weight-Constrained Anisotropic Assignment and Clustering}
\author{Maximilian Fiedler and Peter Gritzmann}

\publishers{\vspace*{4ex}%
	\normalfont\normalsize%
	\parbox{0.8\linewidth}{%
		\textbf{Abstract.} The present paper constructs coresets for weight-constrained anisotropic assignment and clustering. 
		In contrast to the well-studied unconstrained least-squares clustering problem, approximating the centroids of the clusters no longer suffices in the weight-constrained anisotropic case, as even the assignment of the points to best sites is involved. This assignment step is often the limiting factor in materials science, a problem that partially motivates our work.  
		
		We build on a paper by Har-Peled and Kushal, who constructed coresets of size $\mathcal{O}\bigl(\frac{k^3}{\epsilon^{d+1}}\bigr)$ for unconstrained least-squares clustering. We generalize and improve on their results in various ways, leading to even smaller coresets with a size of only $\mathcal{O}\bigl(\frac{k^2}{\epsilon^{d+1}}\bigr)$ for weight-constrained anisotropic clustering. Moreover, we answer an open question on coreset designs in the negative, by showing that the total sensitivity can become as large as the cardinality of the original data set in the constrained case. Consequently, many techniques based on importance sampling do not apply to weight-constrained clustering.
	}	
}
\maketitle

\pagenumbering{arabic}

\section{Introduction}
The present paper constructs coresets of size $\mathcal{O}\bigl(\frac{k^2}{\epsilon^{d+1}}\bigr)$ for weight-constrained  anisotropic least-squares assignment and clustering, for short, \emph{\wcaa} and \emph{\wcac}. 
Our focus is twofold: On the one hand, we improve on known results on coresets for least-squares  clustering and generalize them to \wcac. In this way, we provide new fast approximation algorithms for the NP-hard \wcac problem. On the other hand, we develop our results in line with a specific problem in materials science that requires the computation of \wcaas. While this problem can theoretically be solved in polynomial time, our results allow for significantly reduced computation times and memory, which are strictly limiting factors in practice.

Clustering is a well-established tool for unsupervised learning and has been in wide use for decades in the analysis of data. A variety of more recent applications requires certain \enquote{fairness conditions}  \cite{Huang2019a, Schmidt2020} and, particularly, bounds on the cluster sizes; see, e.g., \cite{Basu2009} for general background on constrained clustering, and the handbook article \cite{GK17} for a concise overview focusing on concepts and results of particular relevance in the present context. For instance, when consolidating farmland \cite{BBG14}, the individual farm sizes should remain near-constant, while electoral districting \cite{Brieden2017} requires balanced districts by law. In conjunction with transformation techniques that extend clustering to supervised learning, bounding the cluster sizes can also drive the output towards statistical significance. Such approaches have recently been used for air cargo prediction \cite{BG20} and response prediction to clinical medication \cite{BG21}.

While computational issues are highly relevant in all such applications, the representation of grain maps, a problem in material science, particularly motivates the present paper \cite{Alpers2015}. This problem concerns structural representations of polycrystalline materials, which are only accessible through specific measurements. In fact, the available grain scan data allows determining the number $k$ of grains, approximations of grain centroids, volumes, and measures of anisotropy. As demonstrated in \cite{Alpers2015}, anisotropic diagrams related to extremal constrained anisotropic clusterings represent grain maps well. In the introduced model, each voxel corresponds to $k$ variables, one for each grain. Hence, a reasonable 3D resolution generally leads to a prohibitively large optimization problem, particularly if grain dynamics requiring solutions for many time steps are studied.

With this and related applications in mind, the present paper addresses basic theoretical questions, including the existence of small coresets for \wcac. 

Coresets, which can be seen as compressed data, have been developed for various problems \cite{Har-Peled2007,Feldman2013,FSS20,Fichtenberger2013,Sohler2018,Bachem2017a}, most notably in our context for approximate unconstrained least-squares clustering. Unfortunately, results for the unconstrained case cannot directly be generalized to the constrained situation, and some concepts do not generalize at all, as we will see in \cref{sec:counterexample}. An apparent difference between the unconstrained and the constrained case is the following: If we assume that the clustering quality is measured with respect to $k$ given points, called \emph{sites}, the unconstrained problem becomes trivial. The points are simply assigned to a nearest site. In contrast, for \wcaa, we need to solve a linear program in $\mathcal{O}(nk)$ variables, which quickly becomes practically intractable for large data sets. In fact, clusterings computed on coresets can only be utilized if fast weight-preserving extensions to clusterings on the full data set are available.

We introduce a general coreset definition that also works in the constrained case and for generalized objective functions. In addition, our coresets have the favorable property that their \wcaas and clusterings can be mapped very quickly to competitive \wcaas and clusterings of the original data set. Hence, after sufficiently good approximate centroids have been determined on coresets, we obtain the desired constrained clustering without having to solve the \enquote{remaining} expensive linear program on the full data set. Thus, our coresets allow us to avoid large-scale computations that would often be virtually impossible for the applications mentioned in materials science. 

As unconstrained least-squares clustering is a special case of \wcac, our results even improve on the previously best-known bounds on coreset sizes for unconstrained least-squares clustering in the standard deterministic setting. In probabilistic models, however, smaller coresets for unconstrained least-squares clustering can be obtained with high probability via \emph{importance sampling}, \cite{Feldman2011g}. It is therefore only natural that \cite{Schmidt2020,Huang2019a} ask whether this technique can also be utilized for constrained clustering. As we will show, however, the total sensitivity of a given instance may be as large as the size of the original data set, which answers this question in the negative.

The present paper is organized as follows. \Cref{sec:notation} introduces the necessary notation, formally specifies our main problems, \wcac and \wcaa, provides relevant background information, and states our main results.  \cref{sec:preliminaries} proves the first results and explains the relation between \wcac and geometric diagrams, as required for our coreset. \Cref{sec:existence_epsilon_1,sec:improved_construction,sec:counterexample} contain the proofs of our main results. Finally, \cref{conclusion} concludes with some remarks and open problems.

\section{Notation, background, and main results}\label{sec:notation}

In the following, we formally introduce the basic problems and state our main results. We begin with a brief introduction to \wcaa and \wcac. Although we provide the most relevant references, we also refer the reader to \cite{GK17} and its sources for pointers to additional literature.

Let $d, n\in \N$ and, for $j\in [n]=\{1,\ldots,n\}$, let $x_j \in \R^d$ and $\omega_{j} \in (0,\infty)$. In the following, we will collect the data in two families $X=\{x_1,\ldots,x_n\}$ and $\Omega=\{\omega_{1},\dots,\omega_{n}\bigr\}$ and refer to $\omega_j$ as the \emph{weight} of the \emph{data point} $x_j$. We refrain from a formal description of the points and weights by means of functions on $[n]$ and tacitly assume that the correspondences are always indicated by the index, i.e. $\omega_j$ is the weight of $x_j$. With this understanding, the pair
$(X,\Omega)$ is called \emph{weighted data set}, and $\omega(X)= \sum_{j=1}^{n} \omega_j$ is its \emph{total weight}.

Further, $k \in \mathbb{N}$ always specifies the \emph{number of clusters} we want to compute. Then, a $k$-\emph{clustering} $C$ of $(X,\Omega)$ is a vector
$$
C = (C_1,\ldots,C_k)=(\xi_{11},\dots, \xi_{1n}, \dots, \xi_{k1}, \dots, \xi_{kn}),
$$ 
where the component $\xi_{ij}$ specifies the fraction of $x_j$ that is assigned to the \emph{$i$th cluster} $C_i=(\xi_{i1},\dots, \xi_{in})$. More formally, the \emph{assignment conditions} are 
\begin{equation*} \label{eq:clustering_constraints}
\xi_{ij}  \ge 0 \quad \forall i \in[k],\ j \in [n] \qquad \text{ and } \qquad \sum_{i=1}^{k} \xi_{ij}  = 1 \quad \forall j \in [n].
\end{equation*}
If $C_i=0$, the $i$th cluster is \emph{void}. 
The \emph{weight} $\omega(C_i)$ and, if $C_i$ is not void, the \emph{centroid} $c_{i}$ of the cluster $C_i$ are given by
$$\omega(C_i)= \sum_{j=1}^{n} \xi_{ij}\omega_j
 \qquad \text{ and } \qquad
c_{i} = c(C_i) = \frac{1}{\omega(C_i)}\sum_{j=1}^{n} \xi_{ij}\omega_j x_{j},
$$
respectively. If $C_i$ is void, we set $c_i=0$. Further, we use the abbreviation $c(C)=\{c_1,\ldots,c_k\}$ for the family of centroids of the clustering $C$, and sometimes we write $c_i=c(C_i)$. The set of all $k$-clusterings of $(X,\Omega)$ will be denoted by $\Cs(k,X,\Omega)$. Whenever the ingredients are clear from the context, we will speak of the elements $C\in \Cs(k,X,\Omega)$ of a clustering of $X$ or simply of a clustering.

As we are particularly interested in weight-constrained clusterings, suppose we are further given lower and upper bounds $\kappa_{i}^{-}, \kappa_{i}^{+}\in \R$ with $\kappa_{i}^{-}\le \kappa_{i}^{+}$ for each cluster, again, collected in a family
$$
    \Kappa = \{\kappa_{1}^{-},\kappa_1^+,\dots,\kappa_{k}^{-},\kappa_k^+\}.
$$ 
% $$
% \Kappa^{-} = \{\kappa_{1}^{-},\dots,\kappa_{k}^{-}\} \qquad \text{ and } \qquad\Kappa^{+} = \{\kappa_{1}^{+},\dots,\kappa_{k}^{+}\}.
% $$ 
Then, we require that the clusters satisfy the \emph{weight constraints} 
\begin{equation*}
\kappa_{i}^{-} \le w(C_i)=\sum_{j=1}^{n} \omega_j \xi_{ij} \le \kappa_{i}^{+} \qquad  \forall i \in [k].
\end{equation*}
The set of all such \emph{weight-constrained clusterings} will be denoted by  $\Cs_\Kappa(k,X,\Omega)$. Note that the simple and easy-to-verify condition
\begin{equation*}
\sum_{i=1}^k \kappa_{i}^{-} \le \omega(X) \le \sum_{i=1}^k\kappa_{i}^{+} 
\end{equation*}
is equivalent to the existence of such clusterings, i.e., to $\Cs_\Kappa(k,X,\Omega) \ne \emptyset$.
To avoid trivialities, we will assume that this condition holds.

Of course, if $\kappa_{i}^{-}=0$, and $\kappa_i^+ \ge \omega(X)$ for an $i\in [k]$ the $i$th weight constraint is redundant and can be omitted. We will signify this situation by writing $\kappa_i^+= \infty$, i.e., formally allowing $\kappa_i^+\in \R\cup \{\infty\}$. Hence, the choice of 
$$\Kappa_\infty = \{0,\infty,\dots, 0,\infty \}$$
refers to the unconstrained case.

Next, we introduce measures for the quality of a given clustering. They are based on distances that may utilize a different norm for each cluster. So, for $i\in [k]$, let  $A_{i}\in \R^{d\times d}$ be a positive definite symmetric matrix, and let $\norm{\ }_{A_i}$ denote the associated \emph{ellipsoidal norm}, i.e.
$\norm{x}_{A_i}=x^TA_ix$ for any $x\in \R^d$. Again, we collect all $k$ such matrices in the family $\Acc=\{A_{1},\dots,A_{k}\}$. Clearly, if 
$E=E_d$ denotes the $d\times d$ unit matrix, $\mathcal{E}=\{E,\dots,E\}$  refers to the classic Euclidean \emph{least-squares} case, in which each cluster norm is $\norm{x}_2=\norm{x}_{E}$. 
Of course, if $\lambda^{-}(A_i)$ and $\lambda^{+}(A_i)$ denote the smallest and largest eigenvalue of $A_i$, then 
$$
\lambda^{-}(A_i)\cdot \norm{x}_2^{2} \le \norm{x}_{A_i}^{2} \le \lambda^{+}(A_i)\cdot\norm{x}_2^{2}.
$$
Such ellipsoidal norms often emerge naturally in many applications. For instance, the different matrices reflect measured knowledge about the moments of the grains in materials, see \cite{Alpers2015}, and enable the representation of anisotropic growth.

We can now measure the quality of a weight-constrained clustering $C \in \Cs_\Kappa(k,X,\Omega)$ with respect to a set of reference points $S=\{s_1,\ldots, s_k \} \subset \R^d$, called \emph{sites}, in terms of 
\begin{equation*}
	\cost_\AK(X,C,S) = 
	\displaystyle \sum_{i=1}^{k}\sum_{j=1}^{n} \xi_{ij} \omega_j \norm{x_{j}-s_{i}}_{A_i}^{2}. 
\end{equation*}
Note that $\cost_\AK(X,C,S)$ only indirectly depends on $\Kappa$, as $C\in \Cs_\Kappa(k,X,\Omega)$ is given. We use the subscript $\Kappa$ consistently in order to indicate that we are considering the constrained case.

Then, \emph{\wcaa} is the following problem:
\begin{quote}
    Given $k, X,\Omega,\Acc,\Kappa, S$, find $C\in \Cs_\Kappa(k,X,\Omega)$ that minimizes $\cost_\AK(X,C,S)$.
\end{quote}
In this setting, the optimization involves only the $nk$ variables $\xi_{ij}$. Hence, in effect, \wcaa is a linear program that can be solved in polynomial time. By our general assumption, $\Cs_\Kappa(k,X,\Omega) \ne \emptyset$, a finite minimum exists, which we denote by $\cost_\AK(X,S)$, i.e.,
\begin{equation*}
	\cost_\AK(X,S) = \min_{C  \in \Cs_\Kappa(k,X,\Omega)} \, \cost_\AK(X,C,S).
\end{equation*}

Instead of fixing the sites and optimizing over the clusterings, we can also fix the clustering, optimize over the sites, and set
\begin{equation*}
    \cost_\AK(X,C) = \min_{S} \, \cost_\AK(X,C,S).
\end{equation*}
In \cref{le:center_replacement}, we will see that the optimal sites are, indeed, the centroids of the clusters. 

Optimizing the latter term over all clusterings results in an optimal weight-constrained anisotropic clustering; its objective function value will be denoted by
\begin{equation*}
    \OPT_\AK(X) = \min_{C \in \Cs_\Kappa(k,X,\Omega)} \, \cost_\AK(X,C).
\end{equation*}
The problem of finding such optimal clustering is called \emph{\wcac} and is defined as:
\begin{quote}
    Given $k,X,\Omega,\Acc,\Kappa$, find $C\in \Cs_\Kappa(k,X,\Omega)$ that minimizes $\cost_\AK(X,C)$.
\end{quote}

As the centroids depend on the clusterings, the objective function of \wcac is highly non-linear in its variables $\xi_{ij}$. It is this property that makes the problem NP-hard and already renders unconstrained least-squares clustering difficult, \cite{Dasgupta2008, Aloise2009}. The latter is shown to be NP-hard even in dimension $d=2$, \cite{Mahajan2012}. Moreover, \cite{Awasthi2015} shows that no PTAS exists in both, the number of clusters $k$, and the dimension of the data $d$. 

These hardness results apply to \wcac, as the unconstrained least-squares clustering problem is contained as the special case $\Acc = \Ecc$ and $\Kappa=\Kappa_\infty$. In this classic situation, we drop $\Kappa$ and $\Acc$ from the above cost terms and simply write $\Cs(k,X,\Omega)=\Cs_{\Kappa_\infty}(k,X,\Omega)$ and
\begin{align*}
    \cost(X,C,S) &= \cost_{\Ecc,\Kappa_\infty}(X,C,S), \quad &\cost(X,S) &= \cost_{\Ecc,\Kappa_\infty}(X,S), \quad  \\
    \cost(X,C) &= \cost_{\Ecc,\Kappa_\infty}(X,C), \quad &\OPT(X) &= \OPT_{\Ecc,\Kappa_\infty}(X).
\end{align*}
If we only drop $\Kappa$ but keep $\Acc$, we are in the anisotropic but unconstrained case, and if we drop $\Acc$ but keep $\Kappa$ we are dealing with the constrained Euclidean case and write, for instance:
\begin{align*}
    \cost_\Acc(X,C,S) &= \cost_{\Acc,\Kappa_\infty}(X,C,S), \quad & \OPT_\Acc(X) &= \OPT_{\Acc,\Kappa_\infty}(X)\\
    \cost_\Kappa(X,C,S) &= \cost_{\Ecc,\Kappa}(X,C,S), \quad& \OPT_\Kappa(X) &= \OPT_{\Ecc,\Kappa}(X).
\end{align*}

A popular method of handling the high computational complexity of unconstrained least-squares clustering is to approximate the data set by much smaller sets that still capture the relevant properties of the data, \cite{Har-Peled2007,Cohen2014,Braverman2016,Makarychev, Chen2009, Feldman2011g}. The decisive property of such a \emph{coreset} is that the cost of a clustering of its points is comparable to the cost of clusterings of the original data set. Hence, the time-consuming computations can be performed on much smaller data sets with good approximations of optimal solutions for the original massive data still being obtained.

We will now define coresets more formally, in a manner that is suitable for the weight-constrained anisotropic case.

\begin{definition}\label{def:coreset}
    Let $k \in \N$, $(X,\Omega)$ and $(\coresetX,\coresetOmega)$ be two weighted data sets, and $\epsilon \in (0,\nicefrac{1}{2}]$, $\delta \in [1,\infty)$.
    
    The weighted data set $(\coresetX,\coresetOmega)$ is an \emph{$(\epsilon,\delta)$-coreset} for $(k,X,\Omega,\Acc,\Kappa)$ if there exists a mapping $\extension:\Cs_\Kappa(k,\coresetX,\coresetOmega) \rightarrow \Cs_\Kappa(k,X,\Omega)$, called an \emph{extension}, and real constants $\Delta^{+},\Delta^{-}$, referred to as \emph{$\Delta$-terms} or \emph{offsets}, with $0\le \Delta^{+} \le \delta \cdot \Delta^{-}$. such that the following two conditions hold for all sets $S$ of $k$ sites and clusterings $\coresetC\in  \Cs_\Kappa(k,\coresetX,\coresetOmega)$:
	\begingroup\leqnomode
	\begin{alignat}{3}
	(1-\epsilon)\cost_\AK(X,\extension(\coresetC),S) &\le \cost_\AK(\coresetX,\coresetC, S) + \Delta^{+},  && \tag{a} \label{def:property_a} \\
	\cost_\AK(\coresetX,S) + \Delta^{-} &\le (1+\epsilon)\cost_\AK(X, S). \;  &&  \tag{b} \label{def:property_b}
	\end{alignat}
	\endgroup	
    If the context is clear, we refer to $(\coresetX,\coresetOmega)$ as an \emph{$(\epsilon,\delta)$-coreset}. If $\delta=1$, we speak of an \emph{$\epsilon$-coreset} and if, in addition, $\Delta^{+}=\Delta^{-}=0$ of a \emph{$0$-offset} or \emph{linear $\epsilon$-coreset}. 
\end{definition}

Note that the conditions in \cref{def:coreset} are required \emph{for every set $S$} of sites. Hence, $(\coresetX,\coresetOmega)$ can also be regarded as a coreset for \emph{each instance} $(k,X,\Omega,\Acc,\Kappa,S)$ of \wcaa, and we will sometimes refer to it this way. 

For easy distinction, we will universally signify coresets by means of a tilde or (if additionally needed) a bar, as in $(\coresetX,\coresetOmega)$ or $(\bar{X},\bar{\Omega})$. Parameters or objects on the different sets will be also be marked this way. For instances, clusterings on $(X,\Omega)$ will usually be named $C$ and their components $\xi_{ij}$, while clusterings and their components on coresets $(\coresetX,\coresetOmega)$ are signified by a tilde, i.e., they are called $\coresetC$ and $\tilde{\xi}_{ij}$.

Before we state our main results, let us briefly comment on the rationale behind the above coreset definition. 

First, note that $\cost_\AK(X,S) \le \cost_\AK(X,\extension(\coresetC),S)$. Hence, in the special situation that 
$\coresetC$ attains $\cost(\coresetX,S)$, 
$\Delta^{+}=\Delta^{-}=0$, and $\delta=1$, the above condition for a linear coreset implies that
$$
(1-\epsilon)\cost_\AK(X,S) \le \cost_\AK(\coresetX,S) \le (1+\epsilon)\cost_\AK(X, S),
$$
which coincides with the coreset definition of \cite{Har-Peled2007} for standard unconstrained least-squares clustering, i.e., $\Acc=\Ecc$, $\Kappa = \Kappa_\infty$. 
Hence, \cref{def:coreset} generalizes the known definition to adapt to the specific requirements of \wcac.

Further, note that the two $\Delta$-terms generalize the one additive offset parameter used in \cite{Feldman2013,FSS20}. They will be used to show that $(\epsilon,\delta)$-coresets preserve $\delta$-approximations up a factor $\epsilon$.

Finally, and most importantly, let us point out that, while rather weak, \cref{def:coreset} still captures the main motivation for the concept that an approximate solution of the constrained clustering problem on a coreset yields an approximate solution on the original data set; see \cref{le:coreset_approximation_preservation_wcaa} for the details.
 
As the general goal is to improve the tractability of the underlying optimization problem we are aiming at coresets $(\coresetX,\coresetOmega)$ of weighted sets $(X,\Omega)$ that are small enough to enhance the performance of available algorithms significantly but still preserve enough of the structure of the original instance to yield solutions which can be extended to good approximations for the original instance. In the unconstrained least-squares case, clusterings of the coreset can easily be converted into clusterings of the original data set. In fact, it suffices to assign the points of $X$ to on of the closest determined sites, and this is efficiently facilitated by means of Voronoi diagrams. However, such a procedure does not respect the weight constraints and, hence, does not produce \wcacs of the full data set in general.

In \cref{def:coreset}, this conversion of a \emph{coreset clustering} $\coresetC \in \Cs_\Kappa(k,\coresetX,\coresetOmega)$ into a clustering $C \in \Cs_\Kappa(k,X,\Omega)$ of the full data set is facilitated by the explicitly introduced extension 
$$
\extension:\Cs_\Kappa(k,\coresetX,\coresetOmega) \rightarrow \Cs_\Kappa(k,X,\Omega).
$$
In fact, $\extension$ captures any way of \enquote{extending} $\coresetC$ to a feasible clustering $C$ of the full data set.
Specifically, condition \cref{def:property_a} of \cref{def:coreset} can be interpreted as saying that for each weight-constrained coreset clustering $\coresetC$, we obtain a weight-constrained clustering $\extension(\coresetC)$ on the full data set of cost that is not much worse. Condition \cref{def:property_b} expresses the property that an optimal \wcaa for sites $S$ on the full data set is not much better than one on the coreset. Hence, good approximations of the latter lead to good approximations of the former.

Let us now turn to the main results of the present paper.
First, as a generalization of a result in \cite{Schmidt2020}, we show that small changes in position of the data points result in $\epsilon$-coresets. The permitted deviation is quantified in terms of the cost $\OPT(X)$ of an optimal unconstrained least-squares clustering on the data set $(X,\Omega)$. Let $\lambdamax$ and $\lambdamin$ denote the largest and smallest eigenvalue of all matrices in $\Acc$.

\begin{theorem}\label{th:movement_error_coreset}
	Let $(k,X,\Omega,\Acc,\Kappa)$ be an instance of \wcac. Further, let  $(\coresetX,\coresetOmega)$ be a weighted data set, $\tilde{n} =|\coresetX|$, and let $p: [n] \rightarrow [\tilde{n}]$ be such that 
	$$
	\coresetomega_{\tilde{\jmath}} =  \sum_{j \in p^{-1}(\tilde{\jmath})} \omega_{j} \qquad \forall \, \tilde{\jmath} \in [\tilde{n}].
	$$
	For any $\epsilon\in (0,\nicefrac{1}{2}]$, if 
	$$ \sum_{j=1}^{n} \omega_{j} \norm{x_{j}-\coresetx_{p(j)}}_{2}^{2} \le \frac{\epsilon^2 \lambdamin}{16 \lambdamax} \OPT(X),
	$$
	then  $(\coresetX,\coresetOmega)$ is a linear $\epsilon$-coreset.
\end{theorem}

This theorem allows us to extend early coreset constructions, namely those that are based on small movements, see, e.g., \cite{Har-Peled2004}, to \wcac. More importantly, \cref{th:movement_error_coreset} is one crucial ingredient of the proof of the following main coreset result.

\begin{theorem}\label{th:smaller_coreset}
	For any instance $(k,X,\Omega,\Acc,\Kappa)$ of \wcac and
	for any $\epsilon\in (0,\nicefrac{1}{2}]$ and 
	$$\delta \ge \frac{\lambdamax}{\lambdamin},$$
	there is an $(\epsilon,\delta)$-coreset of size $\tilde{n}$ with 
	$$\tilde{n} \in \mathcal{O}\left(\frac{k^2}{\epsilon^{d+1}}\right).$$
\end{theorem}

With the $\mathcal{O}$-notation, we refer to the asymptotic behavior for $k\rightarrow \infty$ and $\epsilon \rightarrow 0$. The important property is that the size of the coreset does not depend on the number $n$ of the points of $X$. It does, however, depend exponentially on the dimension $d$ of the data set and involves another factor $\gamma^d$ for some constant $\gamma$. While our proofs allow the growth in $d$ to be specified explicitly, we will refrain from doing so for the sake of simplicity, so that our results are stated for arbitrary but fixed dimension $d$.

The proof of \cref{th:smaller_coreset} follows the basic construction of \cite{Har-Peled2007} but generalizes it to weight-constrained anisotropic clustering and, by suitably exploiting the underlying geometric structure, even reduces the asymptotics by a factor of $k$. In fact, the uniform \enquote{condition number bound} $\delta \ge \lambdamax/\lambdamin$ becomes $1$ for $\Acc=\Ecc$. Hence, we obtain the following corollary for unconstrained least-squares clustering.

\begin{corollary}\label{co:smaller_coreset}
	There is an $\epsilon$-coreset of size $\mathcal{O}\left(\frac{k^2}{\epsilon^{d+1}}\right)$ for least-squares clustering.
\end{corollary}

Finally, we ask whether it is possible to build even smaller coresets using \emph{importance sampling}. This technique is based on a notion of the \enquote{relevance} of points in clusterings, called \emph{sensitivity}, and will be formally defined in \cref{sec:counterexample}. It is well-known that the total sensitivity of unconstrained least-squares clustering can be bounded by $\mathcal{O}(k)$ \cite[Theorem 3.1, for $\alpha=2$]{Langberg2010}, and this results in small corests in the underlying probabilistic model. In the context of \enquote{fair clustering} it was acknowledged in \cite{Huang2019a} as open and actually posed as an \enquote{interesting open question} in \cite{Schmidt2020} whether this approach extends to the constrained case. The following theorem shows, however, that in \wcac, the total sensitivity might actually depend linearly on $n$; hence the answer is in the negative.

\begin{theorem}\label{th:unbounded_sensitivity}
	In the case of \wcac, for any $n \in \Nbb$, there is a data set with $n$ points such that the total sensitivity is $n$.
\end{theorem}

\section{Initial results, auxiliary  lemmas, and diagrams} \label{sec:preliminaries}

In the following, we will first briefly explain the relation between coresets for \wcac and \wcaa. We will then go on to address the power of coresets for constrained clustering, i.e., the question of what kind of error on the original data set can be expected when approximations for \wcac are available on coresets. Finally, we will outline the main coreset construction so as to put the individual parts of the proof of \cref{th:smaller_coreset} into greater perspective. In particular, we will prove a \enquote{composition lemma} for coresets and describe the relation between extremal \wcaas to diagrams. 

Let us begin with observations relating coresets for \wcac to those for \wcaa. Recall that the condition in \cref{def:coreset} is required for every set $S$ of sites. Hence, a coreset for $(\coresetX,\coresetOmega)$ can also be regarded as a coreset for each instance $(k,X,\Omega,\Acc,\Kappa,S)$ of \wcaa. We show that, in a certain sense, the converse is also true, i.e., a coreset according to \cref{def:coreset} satisfies the relevant inequalities when the centroids for the full data set and the coreset, respectively, are used, rather than identical sites in both cases. Before formally establishing this result, we present two basic observations, which are used here and in subsequent sections. 

We emphasize that we mostly apply the following results to a cluster or batch of points using $0$ as an index for the objects, to distinguish them from the full data set. So, let $(X_0,\Omega_0)$ be a weighted data set, let $c_0$ be its centroid, and let $A \in \Rbb^{d \times d}$ be a positive definite symmetric matrix. With $\omega_{x}$ denoting the weight of a point $x\in X_0$, we define the cluster \emph{variation} or batch \emph{error} $\error_{A}(X_0)$ of $X_0$ for $A$ by
\begin{equation*}
\error_{A}(X_0) = \sum_{x \in X_0} \omega_{x} \norm{x-c_0}_{A}^{2}.
\end{equation*}
Note that $\error_{A}(X_0)$ is bounded by $\error_{E}(X_0)$ as follows.

\begin{remark}\label{rem:norm_equivalence}
\begin{equation*} 
	 \lambda_{\min}(A) \error_{\identity}(X_0) \le \error_{A}(X_0) \le \lambda_{\max}(A) \error_{\identity}(X_0).
\end{equation*}
\end{remark}

In accordance with our previous notation, $\omega(X_0) = \sum_{x \in X_0}  \omega_{x}$ denotes the weight of $X_0$. 
The following elementary lemma is a slight generalization of well-known facts. Its first part expresses the weighted sum of squared distances between a site $s$ and the points of a data set $X_0$ as the sum of the variation of $X_0$ and the squared distance between the centroid $c_0$ and $s$. The second part uses this equation to relate \wcac to \wcaa. The proof uses the notion of the \emph{support} $\support(C_{i})$ of the cluster $C_{i}$, defined by 
\begin{equation*}
    \support(C_{i})=\{x_{j} : \ j\in [n] \, \land\, \xi_{ij} > 0\}.
\end{equation*}

\begin{lemma} \label{le:center_replacement}
\begin{enumerate}[label=(\alph*)]
\item Let $(X_0,\Omega_0)$ be a weighted data set, $c_0$ its centroid, $A \in \Rbb^{d \times d}$ a positive definite symmetric matrix, and $s\in \mathbb{R}^{d}$. Then, 
	\begin{equation*}
	\sum_{x\in X_0} \omega_{x} \norm{x-s}_{A}^{2} = \omega (X_0) \cdot\norm{c_{0}-s}_{A}^{2} + \error_{A}(X_0).
	\end{equation*}\label{le:assertion_a}
\item Let $(k,X,\Omega,\Acc,\Kappa,S)$ be an instance of \wcaa and let $C=(C_1,\ldots,C_k)$ be a clustering. Then,
	\begin{equation*}
    \cost_\AK(X,C) \le \cost_\AK(X,C, S)
\end{equation*}\label{le:assertion_b}
with equality if and only if $s_i=c(C_i)$ for every non void cluster $C_i$.
\end{enumerate}
\end{lemma}

\begin{proof} \ref{le:assertion_a} 
Since $\sum_{x\in X_0}\omega_{x}x =\omega(X_0)c_{0}$,
we have
	\begin{align*}
	\sum_{x\in X_0} \omega_{x} \norm{x-s}_{A}^{2} 
	&= \sum_{x\in X_0} \omega_{x} \left(\norm{x-c_{0}}_{A}^{2}+2(x-c_{0})^{\top}A(c_{0}-s) + \norm{c_{0}-s}_{A}^{2} \right) \\
	&=  \omega (X_0)\cdot \norm{c_{0}-s}_{A}^{2} + \error_{A}(X_0)  +2\Bigl( \sum_{x\in X_0}\omega_{x}x-\omega(X_0)c_{0}\Bigr)^{\top}A(c_{0}-s).  \\
	&=  \omega (X_0)\cdot \norm{c_{0}-s}_{A}^{2} + \error_{A}(X_0).
	\end{align*}
\ref{le:assertion_b} Let, as usual, $C=(C_1,\ldots,C_k)$ and define the weighted data set $(X_i,\Omega_i)$ for $i\in [k]$, by 
\begin{equation*}
X_i=\support(C_{i}) \qquad \text{and} \qquad 
\Omega_i=\{\xi_{ij}\omega_j: x_j \in X_i\}.
\end{equation*}
Then, by applying (a) to $(X_0,\Omega_0)=(X_i,\Omega_i)$ and $A=A_i$, for $i\in [k]$, we obtain
    \begin{align*}
        \cost_\AK(X,C,S) &= \sum_{i=1}^k \sum_{j=1}^n \xi_{ij} \omega_j \norm{x_j - s_i}_{A_i}^2 
        =\sum_{i=1}^k \Bigl(\omega(C_i) \cdot\norm{c(C_i)-s_i}_{A_i}^{2} + \error_{A_i}(X_i)\Bigr) \\ 
        &\ge \sum_{i=1}^k \error_{A_i}(X_i) = \cost_\AK(X,C).
    \end{align*}
    Thus, $\cost_\AK(X,C,S) \ge \cost_\AK(X,C)$ with equality if and only if, for each $i\in [k]$, either $\omega(C_i)=0$ or $s_i= c(C_i)$. 
\end{proof}

At first glance, \cref{def:coreset} might seem to require more (all sites) and yield less (cost with respect to the same sites rather than individual centroids) than one might expect. The next corollary shows, however, that the conditions of \cref{def:coreset} imply the corresponding inequalities involving the centroids. 

\begin{corollary}\label{le:wcaa-wcac} 
With the notation of \cref{def:coreset}, let $(\coresetX,\coresetOmega)$ be an \emph{$(\epsilon,\delta)$-coreset} for $(k,X,\Omega,\Acc,\Kappa)$. Then, 
	\begingroup \leqnomode
	\begin{alignat}{3}
	(1-\epsilon)\cost_\AK(X,\extension(\coresetC)) &\le \cost_\AK(\coresetX,\coresetC) + \Delta^{+},  && \tag{a${}^\prime$} \label{def:property_ac} \\
	\OPT_\AK(\coresetX) + \Delta^{-} &\le (1+\epsilon)\OPT_\AK(X). \;  &&  \tag{b${}^\prime$} \label{def:property_bc} 
	\end{alignat}
	\endgroup	
\end{corollary}

\begin{proof}
\cref{def:property_ac}: Let $S=c(\coresetC)$. Then, by \cref{le:center_replacement} \ref{le:assertion_b}
\begin{equation*}
\cost_\AK(X,\extension(\coresetC)) =
\cost_\AK(X,\extension(\coresetC),c(\extension(\coresetC))
\le \cost_\AK(X,\extension(\coresetC),S),
\end{equation*}
and inequality \cref{def:property_a} of  \cref{def:coreset} yields 
\begin{equation*}
    (1-\epsilon)\cost_\AK(X,\extension(\coresetC))
    \le  \cost_\AK(\coresetX,\coresetC,S) + \Delta^{+} = \cost_\AK(\coresetX,\coresetC) + \Delta^{+}.
\end{equation*}
\cref{def:property_bc}: Let $C^*$ be a clustering attaining $\OPT_\AK(X)$ and set $S=c(C^*)$. Then, by \cref{le:center_replacement} \ref{le:assertion_b}
\begin{equation*}
\OPT_\AK(\coresetX) \le \cost_\AK(\coresetX,S),
\end{equation*}
and, hence, by inequality \cref{def:property_b} of  \cref{def:coreset}
\begin{align*}
    \OPT_\AK(\coresetX) + \Delta^{-} \le (1+\epsilon)\cost_\AK(X,S)  =  (1+\epsilon)\OPT_\AK(X),
\end{align*}
which concludes the proof.
\end{proof}

Next, we address the implications of good coreset clusterings on the quality of derived clusterings on the full data set.
As usual, the quality of approximation will be measured in terms of the relative error. Hence, we speak of a \emph{$\gamma$-approximation} if the objective function value of the produced solution is bounded from above by $\gamma$ times its minimum. Before we state and prove the subsequent \cref{le:coreset_approximation_preservation_wcaa}, we deal with the effect of moving sites.

\begin{lemma} \label{le:technical_approximation}
 Let $\instance=(k,X,\Omega,\Acc,\Kappa)$ be an instance of \wcac, $\epsilon \in (0,\nicefrac{1}{2}]$, $\delta \in [1,\infty)$, let $(\coresetX,\coresetOmega)$ be an $(\nicefrac{\epsilon}{3},\delta)$-coreset for $\instance$, and let $\extension$ denote its extension. Further, let $\coresetC\in  \Cs_\Kappa(k,\coresetX,\coresetOmega)$, let $S_1,S_2$ be two sets of $k$ sites, and $\gamma \in \R$ with $\gamma \ge \delta$. Further, suppose that
 \begin{equation*}
     \cost_\AK(\coresetX,\coresetC,S_1) \le \gamma \cost_\AK(\coresetX,S_2).
 \end{equation*}
 Then,   
 \begin{equation*} 
        	\cost_\AK(X,\extension(\coresetC),S_1) \le (1+\epsilon)\gamma\cost_\AK(X,S_2).
    \end{equation*}
\end{lemma}

\begin{proof} By the coreset properties, \cref{def:coreset} \cref{def:property_a} and \cref{def:property_b}, we have
\begin{equation*}
    \left(1-\frac{\epsilon}{3}\right)\cost_\AK(X,\extension(\coresetC),S_1) - \Delta^{+}  \le \cost_\AK(\coresetX,\coresetC,S_1)
\end{equation*}
and
\begin{equation*}
   \cost_\AK(\coresetX,S_2) \le \left(1+\frac{\epsilon}{3}\right)\cost_\AK(X,S_2) -\Delta^{-}.
\end{equation*}
Hence, it follows from the assumption that the right hand side of the former inequality is bounded from above by $\gamma$ times the left hand side of the latter:
\begin{equation*}
    \left(1-\frac{\epsilon}{3}\right)\cost_\AK(X,\extension(\coresetC),S_1) - \Delta^{+}  \le \left(1+\frac{\epsilon}{3}\right)\gamma\cost_\AK(X,S_2) -\gamma\Delta^{-}.
\end{equation*}
Since $\gamma \ge \delta$, and $\Delta^{+} \le \delta\Delta^{-}$, we have $\Delta^{+}\le \delta\Delta^{-} \le \gamma \Delta^{-}$. Hence,
	\begin{align*} 
\left(1-\frac{\epsilon}{3}\right)\cost_\AK(X,\extension(\coresetC),S_1) \le \left(1+\frac{\epsilon}{3}\right)\gamma\cost_\AK(X,S_2).
	\end{align*}
As 
	\begin{equation*}
	\left(1+\frac{\epsilon}{3}\right) = 1 + \frac{2\epsilon}{3} - \frac{\epsilon}{3} \le 1 + \frac{2\epsilon}{3} - \frac{\epsilon^2}{3} = (1+\epsilon)\left(1-\frac{\epsilon}{3}\right),
	\end{equation*}
we obtain 
	\begin{equation*} 
	\cost_\AK(X,\extension(\coresetC),S_1) \le (1+\epsilon)\gamma\cost_\AK(X,S_2),
	\end{equation*}
as asserted.
\end{proof}

Next, we show that approximations on coresets lead to approximations on the full data sets.

\begin{theorem} \label{le:coreset_approximation_preservation_wcaa}
Let $\instance=(k,X,\Omega,\Acc,\Kappa)$ be an instance of \wcac, $\epsilon \in (0,\nicefrac{1}{2}]$, $\delta \in [1,\infty)$, let $(\coresetX,\coresetOmega)$ be an $(\nicefrac{\epsilon}{3},\delta)$-coreset for $\instance$, and let $\extension$ be its extension.

Further, let $\coresetC\in  \Cs_\Kappa(k,\coresetX,\coresetOmega)$, and $\gamma \in \R$ with $\gamma \ge \delta$.
\begin{enumerate}[label=(\alph*)]
    \item If $\coresetC$ is a $\gamma$-approximation for $\coresetinstance_S$, then $\extension(\coresetC)$ is a $(1+\epsilon)\gamma$-approximation for $\instance_S$.\label{le:thm_a}
	 \item
	 If $\coresetC$ is a $\gamma$-approximation for $\coresetinstance$, then $\extension(\coresetC)$ is a $(1+\epsilon)\gamma$-approximation for $\instance$.\label{le:thm_b}
	 \end{enumerate}
\end{theorem}

\begin{proof} \ref{le:thm_a}
	Since $\coresetC$ is a $\gamma$-approximation for $\coresetinstance_S$, we have 
	\begin{equation*}
	\cost_\AK(\coresetX,\coresetC,S) \le \gamma \cost_\AK(\coresetX,S).
	\end{equation*}
	Thus, \cref{le:technical_approximation}, applied with $S_1=S_2=S$, implies that 
	\begin{equation*}
	\cost_\AK(X,\extension(\coresetC),S) \le (1+\epsilon)\gamma\cost_\AK(X,S).
	\end{equation*}
	Hence $\extension(\coresetC)$ is a $(1+\epsilon)\gamma$-approximation for $\instance_S$.
	
\ref{le:thm_b}
	Let $C^*\in  \Cs_\Kappa(k,X,\Omega)$ be optimal and set $S^*=c(C^*)$. By \cref{le:center_replacement} \ref{le:assertion_b},
	\begin{equation*}
	\OPT_\AK(\coresetX) \le \cost_\AK(\coresetX,S^*).
	\end{equation*}
	Hence, together with the assumption that $\coresetC$ is a $\gamma$-approximation for $\coresetinstance$, this yields
	\begin{equation*}
	\cost_\AK(\coresetX,\coresetC,c(\coresetC)) = \cost_\AK(\coresetX,\coresetC) \le \gamma \OPT_\AK(\coresetX) \le \gamma \cost_\AK(\coresetX,S^*).
	\end{equation*}
	Thus, \cref{le:technical_approximation}, applied with $S_1=c(\coresetC)$ and $S_2=S^*$, implies that
	\begin{equation*}
	    \cost_\AK(X,\extension(\coresetC),c(\coresetC)) \le (1+\epsilon)\gamma\cost_\AK(X,S^*)= (1+\epsilon)\gamma\OPT_\AK(X).
	\end{equation*}
    With the aid of \cref{le:center_replacement} \ref{le:assertion_b}, we conclude that
	\begin{equation*}
	\cost_\AK(X,\extension(\coresetC)) \le \cost_\AK(X,\extension(\coresetC),c(\coresetC)) \le (1+\epsilon)\gamma\OPT_\AK(X),
	\end{equation*}
	i.e., $\extension(\coresetC)$ is a $(1+\epsilon)\gamma$-approximation for $\instance$.
\end{proof}

Let us now turn to the third and final part of this section. We will first outline the overall structure of the construction on which our main \cref{th:smaller_coreset} is based, before going on to introduce two important ingredients.

The construction of $(\epsilon,\delta)$-coresets for \wcac follows the basic scheme of \cite{Har-Peled2007} for producing coresets for unconstrained least-squares clustering: First, we compute an $(\alpha,\beta)$-approximation for the least-squares clustering and use the obtained cluster centroids as points, called \emph{vertices}, from which an appropriately dense set of lines is issued. The constructed coreset will \enquote{live} on the union of these \emph{pencils of lines}. Second, we project each point of each cluster on a nearest line issuing from the cluster's centroid. Third, on each line, we merge all points that lie in appropriately constructed batches and collect their weights to obtain a much thinner set $(\coresetX,\coresetOmega)$. The details will be given in \cref{ssec:coreset-construction}.

To show that this three-step scheme works for \wcaa and \wcac, we will make use of the fact that coresets of coresets are themselves coresets. This will be shown in \cref{le:coresets_give_coresets}. Also, in order to verify the coreset properties by bounding the number of costly batches in \cref{le:proof_property_b_bad}, we will utilize the relation between extremal \wcaas and diagrams explained subsequently.

\begin{lemma} \label{le:coresets_give_coresets}
Let $\epsilon_1,\epsilon_2\in (0,\nicefrac{1}{2}]$, $\delta_1,\delta_2 \in [1,\infty)$, and set
\begin{equation*}
    \epsilon=\epsilon_1+\epsilon_2+\epsilon_1\epsilon_2 \quad \text{and} \quad \delta=\max \{\delta_1,\delta_2\}.
\end{equation*} 
Further, let $\instance=(k,X,\Omega,\Acc,\Kappa)$ be an instance of \wcac, let $(\bar{X},\bar{\Omega})$ be an $(\epsilon_1,\delta_1)$-coreset for $\instance$, and let $(\coresetX,\coresetOmega)$ be an $(\epsilon_2,\delta_2)$-coreset 
for $\bar{I}=(k,\bar{X},\bar{\Omega},\Acc, \Kappa)$.
Then, $(\coresetX,\coresetOmega)$ is an $(\epsilon,\delta)$-coreset for $\instance$.
\end{lemma}

\begin{proof}
Let $\extension_1$ and $\extension_2$  denote the extensions for  $(\bar{X},\bar{\Omega})$ and $(\coresetX,\coresetOmega)$, and set $\extension= \extension_1 \circ \extension_2$. Further, for $\ell\in [2]$, let $\Delta_{\ell}^\pm$ denote the corresponding $\Delta$-terms, and set $\Delta^{\pm}  = \Delta_{1}^\pm + \Delta_{2}^\pm$. Since $0\le \Delta_\ell^{+} \le \delta_\ell\Delta_\ell^{-}$ for $\ell\in [2]$, we have 
\begin{equation*}
    0 \le \Delta^{+}  = \Delta_{1}^+ + \Delta_{2}^+ \le \delta_1 \Delta_{1}^- + \delta_2 \Delta_{2}^- \le \delta  \Delta^{-}.
\end{equation*}
Also, note that,
\begin{equation*}
    1-\epsilon=1-(\epsilon_1+\epsilon_2+\epsilon_1\epsilon_2) \le 1-(\epsilon_1+\epsilon_2-\epsilon_1\epsilon_2) = (1-\epsilon_1)(1-\epsilon_2).
\end{equation*}
We will now verify the two conditions of \cref{def:coreset} \cref{def:property_a} and \cref{def:property_b}. So, let $S$ be a set of $k$ sites, and let $\tilde{C} \in  \Cs_\Kappa(k,\coresetX,\coresetOmega)$. Then, by the coreset conditions for $(\bar{X},\bar{\Omega})$ and $(\coresetX,\coresetOmega)$, we have
    \begin{align*}
 (1- \epsilon) &\cost_\AK(X,f(\coresetC),S)
	\le (1-\epsilon_1)(1-\epsilon_2)\cost_\AK(X,f_1(f_2(\coresetC)),S) \\
	& \le (1-\epsilon_2)\left(\cost_\AK(\bar{X},f_2(\coresetC),S) + \Delta_{1}^+\right)
	\le \cost_\AK(\coresetX,\coresetC,S) + \Delta_{2}^+ + (1-\epsilon_2)\Delta_{1}^+ \\
	& \le \cost_\AK(\coresetX,\coresetC,S) + \Delta^+.
	\end{align*}
Further,
\begin{align*}
    \cost_\AK&(\coresetX,S) + \Delta^- 
    \le (1+\epsilon_2)\cost_\AK(\bar{X},S)  + \Delta_{1}^- 
	\le (1+\epsilon_2)\left(\cost_\AK(\bar{X},S)  + \Delta_{1}^- \right) \\
	&\le (1+\epsilon_2)(1+\epsilon_1)\cost_\AK(X,S) 
	= (1+\epsilon)\cost_\AK(X,S). 
\end{align*}
Hence, $(\coresetX,\coresetOmega)$ is an $(\epsilon,\delta)$-coreset for $\instance$.
\end{proof}

Finally, we explain the relation between \wcaas and diagrams.
We will present the relevant theory in a concise manner, which is strongly focussed towards its use in \cref{sec:improved_construction}. For a detailed general account, additional pointers to the literature, and further background information, see \cite{Brieden2017, GK17}.

An \emph{anisotropic diagram}  $\apd=\apd(T,\Sigma,\mathcal{A})$ is specified by a set $T=\{t_{1},\dots,t_{k}\} \subset \mathbb{R}^{d}$ of $k$ different sites,  corresponding \emph{sizes} $\Sigma=\{\sigma_1,\ldots,\sigma_k\} \in \mathbb{R}^k$, and a set of positive definite symmetric matrices $\mathcal{A}=\{A_1,\dots,A_k\}$. It is then defined as the collection $\apd=\{P_1,\ldots,P_k\}$ of subsets $P_i$ of $\R^n$, called \emph{cells}, specified by 
\begin{equation*}
P_{i} = \bigl\{ x\in\mathbb{R}^{d}\ : \,\, \forall  \ell \in [k]: \norm{x-t_{i}}_{A_i}^{2} + \sigma_{i} \le \norm{x-t_{\ell}}_{A_{\ell}}^{2} + \sigma_{\ell} \bigr\}.
\end{equation*}
Since the ellipsoidal norms $\norm{\quad}_{A_i}$ are strictly convex, cells of different sites do not have any interior points in common; see \cite{Brieden2017}. Note that, in the special case of $\Sigma=\{0,\ldots,0\}$ and $\mathcal{A}=\mathcal{E}$, we obtain the classic \emph{Voronoi diagram} of the points in $T$.

Now, let $C=(C_1,\ldots,C_k) \in  \Cs_\Kappa(k,X,\Omega)$. Then, we can say that $\apd$ and $C$ are \emph{compatible} if $\support(C_{i})\subset P_{i} \cap X$ for all $i \in [k]$. More strongly, if $\support(C_{i}) = P_{i} \cap X$, the diagram $\apd$ and the clustering $C$ are called \emph{strongly compatible}. Note that two cells of different sites can only have boundary points in common anyway. As we will see later, the bound on the coreset quality relies on the fact that for relevant clusterings, the number of points that are fractionally assigned to different clusters can be further controlled. More precisely, $C$ and $\apd$ are called \emph{strictly compatible} if $\apd$ and $C$ are strongly compatible and, for each $i,\ell\in [k]$ with $i \ne \ell$ the intersection $(P_i\cap P_\ell)\cap X$ contains at most one point. Later, we will use the following result, which has been shown for the Euclidean case in \cite[Corollary 2.2]{Brieden2012} and which follows from \cite{Brieden2017} in the general case:

\begin{proposition} \label{le:restriction_to_diagrams}
	Given an instance $(k,X,\Omega,\Acc,\Kappa,S)$ of \wcaa, there is a clustering $C$ that attains $\cost_\AK(X,S)$ and admits a strictly compatible anisotropic diagram. Further, such a diagram can be computed in polynomial time.
\end{proposition}

\section{Coresets based on local neighborhood mergers} \label{sec:existence_epsilon_1}
A well-known standard method of constructing coresets for unconstrained least-squares clustering is to merge neighboring data points into a single point that carries their total weight. In this section, we will prove \cref{th:movement_error_coreset}, which shows that a similar approach also works for \wcac. This is an essential first step for the construction of improved coresets in \cref{sec:improved_construction} and, in particular, it allows the repositioning of points to obtain favorable structural properties.

Generally, the merging of points in a neighborhood while preserving weights can be represented by a function that assigns each point of the original data set $(X,\Omega)$ to a point of a new (and typically much smaller) set $(\coresetX,\coresetOmega)$ such that each point of $\coresetX$ carries the sum of the weights of all the original points assigned to it. More formally, with $n=|X|$ and $\tilde{n}=|\coresetX|$ as before, let
\begin{equation*}
    p: [n] \rightarrow [\tilde{n}] 
\end{equation*}
be surjective such that
\begin{equation*}
 \coresetomega_{\tilde{\jmath}} =  \sum_{j \in p^{-1}(\tilde{\jmath})} \omega_{j} \quad \forall \tilde{\jmath} \in [\tilde{n}].
\end{equation*}
Then, we call $p$ a \emph{merging function}. 

If the \enquote{total movement}, i.e., the sum of all distances between a point and its image is small enough, the \enquote{compressed set} $(\coresetX,\coresetOmega)$ is a coreset. For example, using such coreset constructions, an approximate solution for unconstrained least-squares clustering whose cost exceeds $\OPT(X)$ by at most a constant factor can be efficiently found. The existence of such approximation algorithms is stated more formally in the following proposition; see \cite{Aggarwal2009,Kanungo2004a,Har-Peled2007,Feldman2013,FSS20,Braverman2016} for a non-exhaustive list of corresponding algorithms.

\begin{proposition}\label{le:alpha_beta_approximation}
   There exist constants $\alpha \ge 1$, $\beta \in \N$ and a polynomial-time algorithm with the following property: Given an instance $(k,X,\Omega)$ of unconstrained least-squares clustering, the algorithm computes a clustering with at most $\beta k$ clusters with a cost $\ALG_{(\alpha, \beta)}$ satisfying
    \begin{equation*}
        \ALG_{(\alpha, \beta)} \le \alpha \OPT(X).
    \end{equation*}
\end{proposition}

Algorithms according to \cref{le:alpha_beta_approximation} are usually called \emph{$(\alpha,\beta)$-approximations}, as they compute an $\alpha$-approximation but use $k\beta$ rather than $k$ clusters. While, in terms of their approximation properties, smaller constants $\alpha$ and $\beta$ are preferable, the running times typically increase significantly with decreasing constants. Hence, in practice, these effects require some fine-tuning. As pointed out by \cite{Aggarwal2009}, the parameter $\beta$ does not affect the coreset's ability to approximate unconstrained least-squares $k$-clusterings. 

Next, we prove various technical lemmas for handling the weight-constrained anisotropic case. In particular, we will relate the anisotropic case to the Euclidean, bound the cost of an optimal \wcac,  and, finally, define and analyze an appropriate extension function $\extension$.

In order to bound the cost of optimal solutions for a given instance of \wcac, recall first that $\lambdamax$ and $\lambdamin$ denote the largest and smallest eigenvalue of all matrices in  $\Acc$, respectively, and note that the optima in the Euclidean and anisotropic unconstrained situation are related via
\begin{equation*}
	\OPT(X) \le \frac{1}{\lambdamin}\OPT_\Acc(X) \le \frac{\lambdamax}{\lambdamin}\OPT(X).
\end{equation*} 
This shows, in particular, that, for $(\alpha,\beta)$-approximations of least-squares clustering we have the estimate
\begin{equation*}
	\ALG_{(\alpha, \beta)} \le \frac{\alpha}{\lambdamin}\OPT_\Acc(X) \le \frac{\alpha\lambdamax}{\lambdamin}\OPT(X)
\end{equation*} 
for the corresponding anisotropic problem. Also, the former observation implies the following simple but useful consequence.

\begin{lemma}\label{le:euclidean-anisotropic-bound} 
Let $\eta \in (0,\infty)$, and suppose that 
\begin{equation*}
    \sum_{j=1}^n  \omega_j \|\coresetx_{p(j)}- x_j\|^2_{2} \le  \frac{\eta\cdot \lambdamin}{\lambdamax} \OPT(X).
\end{equation*}
Then, whenever $\sum_{i=1}^k\xi_{ij}=1$ for each $j\in [n]$,
\begin{equation*}
    \sum_{i=1}^k \sum_{j=1}^n \xi_{ij} \omega_j \|\coresetx_{p(j)}- x_j\|^2_{A_i} \le \eta \cdot \OPT_\Acc(X).
\end{equation*}
\end{lemma}

\begin{proof} For the proof, just note that
\begin{align*}
    \sum_{i=1}^k  \sum_{j=1}^n & \xi_{ij}\omega_j \|\coresetx_{p(j)}- x_j\|^2_{A_i} 
     \le \lambdamax \sum_{j=1}^n \left( \omega_j \|\coresetx_{p(j)}- x_j\|^2_{2} \cdot \sum_{i=1}^k \xi_{ij}\right) \\
    &= \lambdamax \sum_{j=1}^n  \omega_j \|\coresetx_{p(j)}- x_j\|^2_{2}  
    \le  \eta \lambdamin \OPT(X)
    \le \eta \OPT_\Acc(X).
\end{align*}
\end{proof}

The next technical result considers the effect of merging data points on the cost of clustering. So, let $p$ be a merging function for $(X,\Omega)$, let $(\coresetX,\coresetOmega)$ be the obtained data set and let $C \in \Cs_\Kappa(k,X,\Omega)$. Then, $p$ gives rise to a clustering $\coresetC \in \Cs_\Kappa(k,\coresetX,\coresetOmega)$ defined by
\begin{equation*}
	\tilde{\xi}_{i\tilde{\jmath}} = \frac{1}{\tilde{\omega}_{\tilde{\jmath}}}\sum_{j \in p^{-1}(\tilde{\jmath})} \xi_{ij} \omega_{j} \qquad \forall \, \tilde{\jmath} \in [\tilde{n}], i\in [k].
	\end{equation*}
Slightly abusing notation, we will write $p(X)=\coresetX$,
$p(X,\Omega)=(\coresetX,\coresetOmega)$ and $p(C)=\coresetC$. Then, we obtain the following inequalities.

\begin{lemma} \label{le:technical_bound}
    Let $p$ be a merging function for $(X,\Omega)$, $C \in \Cs_\Kappa(k,X,\Omega)$ and set
    \begin{equation*}
        D=\sum_{i=1}^k \sum_{j=1}^n \xi_{ij} \omega_j \|\coresetx_{p(j)}- x_j\|^2_{A_i}.
    \end{equation*}
Then,
 \begin{align*}
     \cost_\AK(X,&C,S) -  2\sqrt{D\cdot \cost_\AK(X,C,S)} \\
	& \le \cost_\AK(p(X),p(C),S) = \sum_{i=1}^{k} \sum_{j= 1}^{n} \xi_{ij}   \omega_{j} 	\norm{\coresetx_{p(j)} - s_{i}}_{A_i}^{2} \\
	& \le \cost_\AK(X,C,S) + 2\sqrt{D\cdot \cost_\AK(X,C,S)} + D.
	\end{align*}
\end{lemma}

\begin{proof} First, we explain the equality using the definition of $p$:
\begin{align*}
    &\cost_\AK(p(X),p(C),S) = \sum_{i=1}^{k} \sum_{\tilde{\jmath}= 1}^{\tilde{n}} \tilde{\xi}_{i\tilde{\jmath}}   \tilde{\omega}_{\tilde{\jmath}}	\norm{\coresetx_{\tilde{\jmath}} - s_{i}}_{A_i}^{2} \\
    &\quad = \sum_{i=1}^{k} \sum_{\tilde{\jmath}= 1}^{\tilde{n}} \left( \frac{1}{\tilde{\omega}_{\tilde{\jmath}}}\sum_{j \in p^{-1}(\tilde{\jmath})} \xi_{ij} \omega_{j} \right)   \tilde{\omega}_{\tilde{\jmath}}	\norm{\coresetx_{\tilde{\jmath}} - s_{i}}_{A_i}^{2} = \sum_{i=1}^{k} \sum_{\tilde{\jmath}= 1}^{\tilde{n}} \sum_{j \in p^{-1}(\tilde{\jmath})} \xi_{ij} \omega_{j} 	\norm{\coresetx_{\tilde{\jmath}} - s_{i}}_{A_i}^{2} \\
    &\quad = \sum_{i=1}^{k} \sum_{j= 1}^{n} \xi_{ij} \omega_{j} 	\norm{\coresetx_{p(j)} - s_{i}}_{A_i}^{2} .
\end{align*}
Now, note that the latter norm decomposes as follows:
    \begin{align*}
	\norm{\coresetx_{p(j)} - s_{i}}_{A_i}^{2}&
	= \norm{\coresetx_{p(j)} - x_{j} + x_{j} - s_{i}}_{A_i}^{2}\\
	& = \norm{\coresetx_{p(j)}-x_{j}}_{A_i}^{2} +
	2 (\coresetx_{p(j)}-x_{j})^{\top}A_i(x_{j} -s_{i})  +
	 \norm{x_{j}-s_{i}}_{A_i}^{2}.  
\end{align*}
Hence, it suffices to bound the three terms
 \begin{align*}
	& D, \qquad 
	\sum_{i=1}^{k} \sum_{j= 1}^{n} \xi_{ij}   \omega_{j}(\coresetx_{p(j)}-x_{j})^{\top}A_i(x_{j} -s_{i}),\qquad
	\sum_{i=1}^{k} \sum_{j= 1}^{n} \xi_{ij}   \omega_{j}\norm{x_{j}-s_{i}}_{A_i}^{2}
\end{align*}
separately.
The first term appears in the bound from above and, as it is nonnegative, it can be omitted for the estimate from below. Also, the third sum is exactly $\cost_\AK(X,C,S)$. This suffices to suitably bound the absolute value of the second sum. To do so, we evoke the Cauchy-Schwarz inequality twice and obtain
    \begin{equation*} 
    	\begingroup
     	\allowdisplaybreaks
    	\begin{aligned}
    	\Biggl| \sum_{i=1}^{k} \sum_{j=1}^{n} & \xi_{ij} \omega_{j}  (\coresetx_{p(j)}-x_{j})^{\top}A_i(x_{j} -s_{i}) \Biggr| 
        \le  \left| \sum_{i=1}^{k} \sum_{j=1}^{n} \xi_{ij} \omega_{j}  \norm{\coresetx_{p(j)}-x_{j}}_{A_i}\norm{x_{j} -s_{i}}_{A_i} \right|  \\
    	= & \left| \sum_{i=1}^{k} \sum_{j=1}^{n} \sqrt{\xi_{ij} \omega_{j} } \norm{\coresetx_{p(j)}-x_{j}}_{A_i} \sqrt{\xi_{ij} \omega_{j} }\norm{x_{j} -s_{i}}_{A_i} \right| 
    	\le  \sqrt{D} 
    	\cdot \sqrt{\cost_\AK(X,C,S)}.
    	\end{aligned}
    	\endgroup
    \end{equation*}
This concludes the proof.
\end{proof}

As a final ingredient for the proof of \cref{th:movement_error_coreset}, we define a mapping
\begin{equation*}
    \extension: \Cs_\Kappa(k,\coresetX,\coresetOmega) \rightarrow \Cs_\Kappa(k,X,\Omega)
\end{equation*}
that converts clusterings of $(\coresetX,\coresetOmega) = p(X,\Omega)$ to clusterings of $(X,\Omega)$. In fact, $\extension$ assigns each point $x_j$ of $X$ to clusters exactly as its representative $x_{p(j)}$ in $\coresetX$ is assigned. More precisely, let $\coresetC \in \Cs_\Kappa(k,\coresetX,\coresetOmega)$ and define $C=\extension(\coresetC)$ by setting 
\begin{equation*}
\xi_{ij} = \tilde{\xi}_{ip(j)} \qquad \forall \, j \in [n], i\in [k].
\end{equation*}
The following simple lemma shows, in particular, that $\extension$ preserves the cluster weights, i.e., $\extension(\coresetC) \in \Cs_\Kappa(k,X,\Omega)$.  

\begin{lemma} \label{le:technical_bound-f}
Let $p$ be a merging function for $(X,\Omega)$, $(\coresetX,\coresetOmega)=p(X,\Omega)$, $\coresetC \in \Cs_\Kappa(k,\coresetX,\coresetOmega)$, and $\extension$ as defined above. Then, for $C = f(\coresetC)$,
\begin{enumerate}[label=(\alph*)]
    \item $\displaystyle \sum_{j=1}^n \xi_{ij} \omega_{j} =\sum_{\tilde{\jmath}=1}^{\tilde{n}} \tilde{\xi}_{i\tilde{\jmath}}\tilde{\omega}_{\tilde{\jmath}}$ for each $i\in [k]$;\label{le:item-a}
\item $p \bigl(f(\coresetC)\bigr)= \coresetC$. \label{le:item-b}
\end{enumerate}
\end{lemma}

\begin{proof} For the proof of \ref{le:item-a} simply note that
\begin{equation*} 
\sum_{j=1}^n \xi_{ij} \omega_{j} = \sum_{j=1}^n \tilde{\xi}_{ip(j)} \omega_{j} = \sum_{\tilde{\jmath}=1}^{\tilde{n}} \sum_{j \in p^{-1}(\tilde{\jmath})} \xi_{ij} \omega_{j} = \sum_{\tilde{\jmath}=1}^{\tilde{n}} \tilde{\xi}_{i\tilde{\jmath}}\tilde{\omega}_{\tilde{\jmath}}.
\end{equation*}
To prove \ref{le:item-b}, denote the components of $p\bigl(f(\coresetC)\bigr)$ by $\tilde{\xi}_{i\tilde{\jmath}}^*$. Then, we have
\begin{equation*} 
	\tilde{\xi}_{i\tilde{\jmath}}^* = \frac{1}{\tilde{\omega}_{\tilde{\jmath}}}\sum_{j \in p^{-1}(\tilde{\jmath})} \xi_{ij} \omega_{j} = 
	\frac{\tilde{\xi}_{i\tilde{\jmath}}}{\tilde{\omega}_{\tilde{\jmath}}}\sum_{j \in p^{-1}(\tilde{\jmath})} \omega_{j}= \tilde{\xi}_{i\tilde{\jmath}}.
\end{equation*}
This concludes the proof of the lemma.
\end{proof}

By \cref{le:technical_bound-f} \ref{le:item-a}, $\extension$ preserves the cluster sizes and maps weight-constrained clusterings onto weight-constrained clusterings. Also note, that $\extension$ can be very easily computed. 

\cref{le:technical_bound-f} \ref{le:item-b} allows us to infer that the inequalities of \cref{le:technical_bound} not only hold as stated, i.e., for pairs $C \in \Cs_\Kappa(k,X,\Omega)$ and $\coresetC=p(C)$, but also for pairs $\coresetC \in\Cs_\Kappa(k,\coresetX,\coresetOmega)$ and $C=f(\coresetC)$, i.e., 

\begin{corollary}\label{le:technical_bound-2}
Let $p$ be a merging function for $(X,\Omega)$, $(\coresetX,\coresetOmega)=p(X,\Omega)$, and $\coresetC \in\Cs_\Kappa(k,\coresetX,\coresetOmega)$. Then, with $D$ as before,
\begin{align*}
     \cost_\AK(X,f(\coresetC),S) &-  2\sqrt{D\cdot \cost_\AK(X,f(\coresetC),S)} 
	 \le \cost_\AK(\coresetX,\coresetC,S)\\
	 &\le \cost_\AK(X,f(\coresetC),S) + 2\sqrt{D\cdot \cost_\AK(X,C,S)} + D.
	\end{align*}
\end{corollary} 

We are now ready to prove \cref{th:movement_error_coreset}.

\begin{proof}[Proof of \cref{th:movement_error_coreset}]
We show that, under the assumptions of the theorem, the coreset properties, \cref{def:coreset} \cref{def:property_a}, \cref{def:property_b}, hold with offset $0$, i.e., for $\Delta^{+}=\Delta^{-}=0$.

Since $\OPT_\Acc(X) \le \cost_\AK(X,C,S)$ and, by assumption,
\begin{equation*}
    \sum_{j=1}^{n} \omega_{j} \norm{x_{j}-\tilde{x}_{p(j)}}_{2}^{2} \le \frac{\epsilon^2 \lambdamin}{16 \lambdamax} \OPT(X),
\end{equation*}
\cref{le:euclidean-anisotropic-bound}, applied with $\eta= \nicefrac{\epsilon^2}{16}$, provides 
    \begin{equation*}
        D=\sum_{i=1}^k \sum_{j=1}^n \xi_{ij} \omega_j \|\coresetx_{p(j)}- x_j\|^2_{A_i} \le \frac{\epsilon^2}{16} \OPT_\Acc(X) \le  \frac{\epsilon^2}{16}\cost_\AK(X,C,S).
    \end{equation*}
Thus, \cref{le:technical_bound-2} yields for $\coresetC \in\Cs_\Kappa(k,\coresetX,\coresetOmega)$ and $C=f(\coresetC)$
\begin{align*}
(1-\epsilon)\cost_\AK(X,C,S)  &\le  \cost_\AK(X,C,S) -  2\sqrt{D\cdot \cost_\AK(X,C,S)}\\
&\le \cost_\AK(\coresetX,\coresetC,S),
\end{align*}
which shows coreset property \cref{def:property_a}.

To verify coreset property \cref{def:property_b}, let $C \in \Cs_\Kappa(k,X,\Omega)$ be optimal for the given sites $S$, and set $\coresetC=p(C)$. Then, by \cref{le:technical_bound}
\begin{align*}
\cost_\AK(\coresetX,S) &\le \cost_\AK(\coresetX,\coresetC,S)
\le \cost_\AK(X,C,S) + 2\sqrt{D\cdot \cost_\AK(X,C,S)} + D\\
& \le \left(1+\frac{\epsilon}{2}+\frac{\epsilon^2}{16}\right)\cost_\AK(X,S) \le (1+\epsilon)\cost_\AK(X,S).
\end{align*}
This concludes the proof.
\end{proof}

\section{Small coresets} \label{sec:improved_construction}
We will now construct small $(\epsilon,\delta)$-coresets, analyze their properties and prove \cref{th:smaller_coreset}. In \cref{ssec:coreset-construction}, we will first provide the construction and bound the sizes of the resulting coresets $(\coresetX,\coresetOmega)$. In \cref{ssec:interval_partition}, we will derive a structural property for clusterings on $(\coresetX,\coresetOmega)$ that admit a strictly compatible anisotropic diagram. Finally, the coreset properties according to \cref{def:coreset} will be verified in  \cref{ssec:coreset-properties}.

\subsection{The construction}\label{ssec:coreset-construction}
The construction will follow the 3-step scheme already outlined in \cref{sec:preliminaries}. The first step is to compute an $(\alpha,\beta)$-approximation for least-squares clustering using \cref{le:alpha_beta_approximation}. This approximation results in an integer assignment of the points of $X$ to $\hat k$ clusters $\hat C_i$ with $\hat k \le \beta\cdot k$ of cost
\begin{equation*}
    \text{ALG}_{(\alpha,\beta)}(X)\le \alpha \cdot \text{OPT}(X).
\end{equation*}
For $i\in [\hat k]$, let $\hat c_i$ denote the centroid of $\hat C_i$. (The following exposition is written with the general situation in mind that none of the clusters is void. Note, however, that void clusters do not cause any problems but actually reduce the size of the constructed coreset.)

We start by describing the second step in full detail. In this step, the obtained cluster centroids $\hat c_i$ are used as vertices, from which suitable pencils of lines are issued. This follows the approach of \cite{Har-Peled2007}, which essentially reduces the coreset construction in the least-squares case to dimension one. In order to construct such pencils, we employ \emph{$\epsilon$-nets}, i.e., point sets $Q$ with properties explained in the following lemma. 

\begin{proposition}[see \cite{Matousek2002,Har-Peled2007}] \label{le:epsilon_nets}
Let, as usual, $\mathbb{S}^{d-1}$ denote the Euclidean unit sphere of $\R^d$. Given $\epsilon_0 \in (0,\nicefrac{1}{2})$, there exists a point set $Q=Q(\epsilon_0) \subset \mathbb{Q}^d\setminus \{0\}$ with the following properties:
    \begin{enumerate}[label=(\alph*)]
        \item $Q$ contains $\mathcal{O}(\epsilon_0^{-(d-1)})$ points.
        \item For each $p\in \mathbb{S}^{d-1}$ there exists a point $q \in Q$ such that $\norm{p-q}_2 \le \epsilon_0$, and
        \item $Q$ can be computed in time  $\mathcal{O}(\epsilon_0^{-(d-1)})$.
    \end{enumerate}
\end{proposition}
Using \cref{le:epsilon_nets}, we compute such a set $Q(\epsilon_0)$ for 
\begin{equation*}
    \epsilon_0=\frac{\epsilon}{4}\sqrt{\frac{\lambdamin}{\alpha \cdot \lambdamax}}.
\end{equation*}
This leads to $\hat k$ pencils $\hat c_i + \R Q$ of the lines 
\begin{equation*}
  L(\hat c_i,q) = \hat c_i + \R q, \qquad \forall i\in [\hat k], q\in Q(\epsilon_0).
\end{equation*}
Note that, as $0\not\in Q$, all sets $L(\hat c_i,q)$ are indeed lines. Let $\mathcal{L}$ denote the set of all such lines. Then,
\begin{equation*}
    |\mathcal{L}| = \hat k \cdot |Q| \in \mathcal{O}(k \epsilon^{-(d-1)}).
\end{equation*}

Next, for each $i\in [\hat k]$, we project each point of $x\in X\cap \hat C_i$ orthogonally on a closest ray $L(\hat c_i,q)$; ties are resolved arbitrarily. The resulting point will be denoted by $\bar{x}$. Note that, unless $Q$ is chosen in an appropriately general position, it may happen that two points of $X$ are projected onto the same point $\bar{x}$. We do not merge such points, and we treat them as separate entities in the family $\bar{X}$. The weights then remain unchanged, and we obtain the weighted data set $(\bar{X}, \Omega)$ of the same cardinality $n$, i.e., the merging function is the identity. As the following lemma shows, $(\bar{X}, \Omega)$ is a $0$-offset coreset for \wcac, with the extension $\extension_1: \Cs_\Kappa(k,\bar{X},\Omega) \rightarrow \Cs_\Kappa(k,X,\Omega)$ defined by 
\begin{equation*}
\xi_{ij} = \bar{\xi}_{ij} \qquad \forall \, j \in [n], i\in [k].
\end{equation*}

\begin{lemma} \label{le:projecting_is_a_coreset}
    $(\bar{X}, \Omega)$ is a linear $\epsilon$-coreset for \wcac.
\end{lemma}

\begin{proof} Let $j\in [n]$ and $i(j)\in [\hat k]$ such that $x_j\in X\cap \hat C_{i(j)}$. Further, let $q\in Q$ such that $\bar{x}_j \in L(\hat c_{i(j)},q)$.
Using the intercept theorem, we see that with $i=i(j)$,
    \begin{equation*}
        \norm{x_j-\bar{x}_j}_2 \le 
        \left\| \left(\hat c_i + \frac{x_j-\hat c_i}{\norm{x_j-\hat c_i}_2}\right) - (\hat c_i +q) \right\|_2 \cdot \norm{x_j-\hat c_i}_2 \le \epsilon_0 \cdot \norm{x_j-\hat c_i}_2.
    \end{equation*}
Since the $\hat k$ clusters $\hat C_i$ constitute an $(\alpha,\beta)$-approximation for least-squares clustering, this implies   
    \begin{align*}
        \sum_{j=1}^{n} \omega_{j} \norm{x_{j}-\bar{x}_{j}}_{2}^{2} 
        &\le \epsilon_0^2 \cdot \sum_{j=1}^{n} \omega_{j}  \norm{x_j-\hat c_{i(j)}}_2^2 
        = \epsilon_0^2 \cdot \text{ALG}_{(\alpha,\beta)}(X)
        \le \epsilon_0^2 \cdot \alpha \cdot  \OPT(X)\\
        &= \frac{\epsilon^2 \lambdamin}{16 \lambdamax } \OPT(X).
    \end{align*}
Hence, it follows from \cref{th:movement_error_coreset} that $(\bar{X},\Omega)$ is an $\epsilon$-coreset for \wcac.
\end{proof}

Let us point out that together, Lemmas \ref{le:coresets_give_coresets} and \ref{le:projecting_is_a_coreset} imply that it suffices to design coresets that live on the union of the constructed pencils, i.e. on $\mathcal{L}$. This describes the second step of the procedure.

In the third step, we merge points on each line into batches by replacing them with their centroids and adjusting their weights. This, finally, results in the desired data set $(\coresetX,\coresetOmega)$. 

So, let $L \in \mathcal{L}$. Starting from the leftmost point, we successively add, from left to right, points of $\bar X\cap L$ to form the first batch $B=B_{1}$ until 
\begin{equation*}
    \error_{\identity}(B)= \error_{\identity}=\frac{\epsilon^2}{32\alpha k \lambdamax |\mathcal{L}|}\text{ALG}_{(\alpha,\beta)}(X),
\end{equation*}
and continue with the next batch. The above condition can always be achieved by assigning, if needed, the last point fractionally to the batch. In such a case, we may assume that all points are completely assigned to a single batch by splitting this point into two with the appropriate weights. While this assumption simplifies the exposition, it does not effect on the results and can be made without loss of generality.

The process is continued until all points are assigned to batches. Note that all batches have the same error, apart from the last one, whose error may be smaller. Let $\batches(L)$ denote the set of batches on the line $L$, and let $\batches$ signify the total set of all batches, i.e., the union of all $\batches(L)$ with $L \in \mathcal{L}$.

Finally, we merge the points of each batch. More precisely, let $B \in \batches$, and let
\begin{equation*}
\omega_B = \sum_{\bar{x}_j \in \bar X\cap B} \omega_j, \qquad
    x_B= \frac{1}{\omega_B } \sum_{\bar{x}_j \in \bar{ X}\cap B} \omega_j \bar{x}_j
\end{equation*} 
be the set's $\bar{ X}\cap B$ total weight and centroid, respectively.
We replace the points of $\bar X\cap B$ by their centroid $x_B$ and assign $\omega_B$ as its weight. 

The data set that results from applying the merger to all lines of the pencils will be called $(\batchingX,\batchingOmega)$. This is actually the desired set $(\coresetX,\coresetOmega)$, and we will use both notations
interchangeably, depending on which one leads to the more intuitive exposition in the subsequent proofs. In particular, the number of points will still be denoted by $\tilde{n}$. If we index the batches as $B_{\tilde{\jmath}}$ for $\tilde{\jmath} \in [\tilde{n}]$ and set $x_{\tilde{\jmath}}=x_{B_{\tilde{\jmath}}}$, the corresponding merging function $p:[n]\rightarrow [\tilde{n}]$ is again specified by
\begin{equation*}
    p(j)= \tilde{\jmath} \quad \Longleftrightarrow \quad 
    \bar{x}_j \in B_{\tilde{\jmath}},
\end{equation*}
and the extension $\extension_2: \Cs_\Kappa(k,\coresetX,\coresetOmega) \rightarrow \Cs_\Kappa(k,\bar{X},\Omega)$ is again defined by 
\begin{equation*}
\bar{\xi}_{ij} = \tilde{\xi}_{ip(j)} \qquad \forall \, j \in [n], i\in [k].
\end{equation*}
Hence, the extension $\extension: \Cs_\Kappa(k,\coresetX,\coresetOmega) \rightarrow \Cs_\Kappa(k,X,\Omega)$  is simply given by
\begin{equation*}
f= f_1\circ f_2, \quad \text{i.e.,} \quad \xi_{ij} = \tilde{\xi}_{ip(j)} \qquad \forall \, j \in [n], i\in [k].
\end{equation*}

This completes the construction. Before we prove that $(\batchingX,\batchingOmega)$ is an $(\epsilon,\delta)$-coreset, we bound its cardinality $\tilde{n}$ in this section. The basic idea for deriving a bound is due to \cite{Har-Peled2007}. However, we improve on their arguments to reduce the dependency of the estimate by a factor of $k$. 

\begin{lemma} \label{le:number_of_points_batching}
\begin{equation*}
    |\batches| \in \mathcal{O}\left(\frac{k^2}{\epsilon^{d+1}}\right),
\end{equation*}
i.e., the number of points in $\batchingX$ is of the order $\nicefrac{k^2}{\epsilon^{d+1}}$.
\end{lemma}

\begin{proof} To establish the asserted bound, let $\bar{c}_1,\ldots,\bar{c}_k$ be the centroids of an optimal least-squares clustering of $(\bar{X},\Omega)$. In this unconstrained Euclidean case, the clustering is compatible with a Voronoi diagram. As we need not respect any weight constraints for the clusters, we may assume that no point of $\bar{X}$ lies on the boundary of any Voronoi cell. Since $L$ intersects each Voronoi cell in a (possibly empty or one-element) interval, all points in the same interval lie in the same cluster. This implies that for at most $k-1$ batches in $L$, the points are assigned to more than one cluster, i.e., at least $|\batches(L)|-k+1$ such batches are fully assigned to a single cluster. 

Now, let $B\in \batches(L)$ be contained in the $i$th cluster, but let $B$ not be the last batch on $L$. As the error $\error_{E}(B)$ is known from the construction, we obtain from \cref{le:center_replacement} (with $X_0=B$, $c_0=c_B$, $\omega(X_0)=\omega_B$ and $A=E$) that
	\begin{align*}
	\sum_{\bar{x}\in B} \omega_{\bar{x}} \norm{\bar{x}-\bar{c}_i}_{2}^{2} &= \omega_{B}\norm{c_{B}-\bar{c}_i}_{2}^{2} + \error_{E}(B) \ge \error_{E}(B)= \error_{\identity}.
	\end{align*}
As at most $(k-1)+1=k$ of the batches $B\in \batches(L)$ can contribute less than $\error_{\identity}$ to the cost $\text{OPT}(\bar{X})$ of the clustering, we see that 
\begin{equation*}
  |\batchingX| \le \frac{\text{OPT}(\bar{X})}{\error_{\identity}} + k\cdot |\mathcal{L}|.
\end{equation*}
Since $(\bar{X},\Omega)$ is a linear $\epsilon$-coreset for $(X,\Omega)$ by \cref{le:projecting_is_a_coreset}, we can use \cref{le:wcaa-wcac} (b$'$) to obtain
\begin{equation*}
\text{OPT}(\bar{X}) \le (1+\epsilon) \cdot \text{OPT}(X) \le 2 \cdot \text{ALG}_{(\alpha,\beta)}(X).
\end{equation*}
All in all,
	\begin{equation*}
|\batchingX| \le 2\cdot \frac{\text{ALG}_{(\alpha,\beta)}(X)}{\error_{\identity}} +  k\cdot |\mathcal{L}| = \Bigl(64 \cdot \alpha \cdot \lambdamax \cdot\frac{k}{\epsilon^{2}} + k\Bigr)  \cdot |\mathcal{L}|.
	\end{equation*}
	As the right term is in $ \mathcal{O}\left(\frac{k^2}{\epsilon^{d+1}}\right)$, this concludes the proof.
\end{proof}

\subsection{Structural properties of clusterings admitting strictly compatible diagrams}  \label{ssec:interval_partition}

We will now show that each clustering that admits a strictly compatible diagram gives rise to a partition of each line $L$ into at most $2k-1$ intervals such that no data point in the relative interior of each interval is fractionally assigned, i.e., belongs to more than one cluster. 
Since by \cref{le:restriction_to_diagrams} it suffices to consider such diagrams on the coresets, this will allow us to verify coreset condition  \cref{def:property_b} of \cref{def:coreset}. Condition \cref{def:property_a}, on the other hand, does not require additional assumptions on the structure of the clustering. 

Given an instance $(k,\coresetX,\coresetOmega,\Acc,\Kappa, S)$ of \wcaa, suppose that $\coresetC\in \Cs_\Kappa(k,\coresetX,\coresetOmega)$ attains $\cost(\coresetX,S)$ and admits a strictly compatible diagram $\apd=\apd(T,\Sigma,\mathcal{A})=\{P_1,\ldots,P_k\}$. Note that $\coresetC$ is available by  \cref{le:restriction_to_diagrams}. 

We use $\apd$ to bound the \enquote{structural complexity} of $\coresetC$'s intersections with lines $L$. Later we apply the results to those lines that constitute the pencils on which $(\coresetX,\coresetOmega)$ \enquote{lives}. 

The first lemma considers the \enquote{generic situation} and could be stated within the realm of \emph{Davenport-Schinzel Sequences}; see, e.g., \cite{SA95}. We will, however, give a direct, self-contained proof of the relevant result.

So, let $c\in \R$, $q\in \R^d\setminus \{0\}$, $L=L(c,q)$, and, for $i\in [k]$,
\begin{equation*}
    g_i:\R^d \rightarrow \R, \qquad g_i(x)= \|x-s_i\|^2_{A_i} + \sigma. 
\end{equation*}
Further, let
\begin{equation*}
    h:\R^d \rightarrow \R,  \qquad h(x)= \min \{g_i(x): i\in [k]\}
\end{equation*}
be the lower envelope of $g_1,\ldots,g_k$. 
Then, of course, the diagram $\apd$ is induced by $g_1,\ldots,g_k$, and 
\begin{equation*}
    L\cap P_i= \{x\in L: g_i(x)=h(x)\}
\end{equation*}
consists of finitely many  connected components, all closed intervals and \emph{proper}, i.e., $1$-dimensional or possibly singletons. Let $\gamma_i$ denote their number. Then, we have the following bound.

\begin{lemma}\label{thm:intersection}
Suppose that no two of the restrictions $g_i|_L$ coincide. Then,
\begin{equation*}
    \sum_{i=1}^k \gamma_i \le 2k-1,
\end{equation*}
and this bound is the best possible.
\end{lemma}

\begin{proof} 
The proof of the inequality is by induction with respect to $k$. As the assertion is trivial for $k_0=1$, suppose that it has been verified for some $k_0\in \N$, and let $k=k_0+1$.

Suppose, first, that there exists an index $i^*\in [k]$  such that $\gamma_{i^*}=1$. By the induction hypothesis, we have
$\sum_{i\in [k]\setminus \{i^*\}} \gamma_i \le 2(k-1)-1$, and adding $g_{i^*}$ back in, the corresponding interval can split at most one other interval by the assumption that no two of the restrictions $g_i|_L$ coincide. Hence, the assertion holds.

So, suppose in the following,  that $\gamma_{i}\ge 2$ for all $i\in [k]$. We will show that this leads to a contradiction, thereby proving the theorem.

For $i\in [k]$, let $y_j^{(i)}$ be a relative interior point in the $j$th component of $L\cap P_i$, ordered such that
$y_1^{(i)} < \ldots < y_{\gamma_i}^{(i)}$.
Further, let
\begin{equation*}
    i_0 \in M_1= \argmax \bigl\{y_1^{(i)}: i\in [k]\bigr\}.
\end{equation*}
If the index $i_0$ is unique, set $i^*=i_0$. Otherwise let $i^* \in M_2= \argmin \bigl\{y_2^{(i)}: i\in M_1\bigr\}$.
Note that $M_2$ still does not need to be a singleton, but the different choices do not affect the subsequent arguments.

Since $\gamma_{i^*}\ge 2$, there must be $i\in [k]\setminus \{i^*\}$ and $j \in [\gamma_i]$ such that
\begin{equation*}
    y_1^{(i^*)} < y_j^{(i)}  < y_2^{(i^*)} .
\end{equation*}
As $i^* \in M_1$, we know that $j\ge 2$ and $y_1^{(i)}  \le y_1^{(i^*)}$. Equality, however, is excluded by the choice of $i^* \in M_2$. Hence,
\begin{equation*}
    y_1^{(i)} < y_1^{(i^*)} < y_j^{(i)}  < y_2^{(i^*)}.
\end{equation*}
This implies that $g_i-g_{i^*}$ has at least three roots in $L$, i.e., $g_i|_L=g_{i^*}|_L$, which contradicts our general assumption.

The bound is tight for examples of suitably nested parabolas.
\end{proof}

As a consequence, we obtain the following result.

\begin{corollary} \label{th:generalized_structure-generic}
Suppose again that no two of the restrictions $g_i|_L$ coincide.  
Then, there is a dissection of $L$ into at most $2k-1$ proper intervals such that the relative interior of each interval is contained in exactly one of the cells of $\apd$.
\end{corollary}

\begin{proof}
For $i\in [k]$, let $\mathcal{I}_i=\mathcal{I}_i^{(0)}$ denote the set of the relative interiors of all proper intervals that occur as connected components of $L\cap P_i$, and set $\mathcal{I}=\mathcal{I}^{(0)}=\mathcal{I}_1 \cup \ldots \cup \mathcal{I}_k$. Since, by assumption, all univariate polynomials $g_i|_L$ are different, any two intervals in $\mathcal{I}$ are disjoint.

For $i\in [k]$, let $\gamma_i^{(0)}$ denote the number of singletons that occur as connected components of $L\cap P_i$, and set $\gamma_i^{(1)}=|\mathcal{I}_i|$. Then, of course,
\begin{equation*}
 \gamma_i=\gamma_i^{(0)}+\gamma_i^{(1)} \quad \forall i\in [k], \qquad \text{and}\qquad |\mathcal{I}|=\sum_{i=1}^k \gamma_i^{(1)}.
\end{equation*}
Starting with $\mathcal{I}^{(0)}$, we will now construct refinements $\mathcal{I}^{(\ell)}$ by successively splitting intervals that violate the assertion as long as such intervals exist.

So, let $i_0\in [k]$, let $I_0\in \mathcal{I}_{i_0}$, and suppose that there exists a point $y\in I_0$ that also belongs to some other cell $P_{i_1}$. 
Then, $y$ is the unique point of $P_{i_0}\cap P_{i_1}$ in $\relint (I_0)$ since, otherwise, $g_{i_0}|_L- g_{i_1} |_L$ would have at least three local extrema, contradicting the assumption that $g_{i_0}|_L\ne g_{i_1} |_L$. For the same reason, no other interval $I\in \mathcal{I}_{i_0}$ can contain a point of $P_{i_1}$ in its relative interior.

Now, we split $I_0$ at $y$ into two intervals and obtain a new set $\mathcal{I}^{(1)}$ by replacing $I^{(0)}$ in $\mathcal{I}_{i_0}$ by the corresponding two relatively open intervals $I_-^{(0)}$ and $I_+^{(0)}$. Note that neither $I_-^{(0)}$ nor $I_+^{(0)}$ contains points of $P_{i_1}$ anymore, and that $\{y\}$ is a connected component of $L\cap P_{i_1}$, which counted towards $\gamma_{i_1}^{(0)}$. 

We continue by successively splitting intervals that violate the assertion and finally obtain a set 
$\mathcal{I}^*$ of intervals that have the desired property.
The total number of necessary splits is bounded by $\sum_{i=1} \gamma_i^{(0)}$. Hence, with the aid of \cref{thm:intersection},
\begin{equation*}
    |\mathcal{I}^*|\le \sum_{i=1} \gamma_i^{(1)} + \sum_{i=1} \gamma_i^{(0)} = \sum_{i=1} \gamma_i \le 2k-1,
\end{equation*}
which is the asserted bound.
\end{proof}

As an immediate consequence, the corollary yields a structural result when $\coresetC$ and $\apd$ are compatible.  

\begin{corollary} \label{th:generalized_structure-generic-clustering}
Let $g_1|_L,\ldots,g_k|_L$ be all different, and let $\coresetC$ and $\apd$ be compatible. Then, there exists a partition of $L$ into at most $2k-1$ proper intervals such that all data points in the relative interior of each interval are fully assigned to the same cluster. 
\end{corollary}

Note, however, that the \enquote{genericity assumption} that the $k$ univariate polynomials $g_1|_L,\ldots,g_2|_L$ are all different for each line $L$ can generally not be guaranteed, even if all multivariate polynomials $g_1,\ldots,g_k$ are different. For instance, when $\mathcal{A}=\mathcal{E}$, the cells of $\apd$ are polyhedra that might all have a line (or even some higher-dimensional affine subspace) in common. In this case, it may happen that all points of $\tilde{X}$ in the relative interior of a component $I$ of $L\cap P_i$ are fractionally assigned to $C_i$; see \Cref{fig:s_compatible_does_not_suffice} for an illustration.
\begin{figure}
    \centering
    \includegraphics{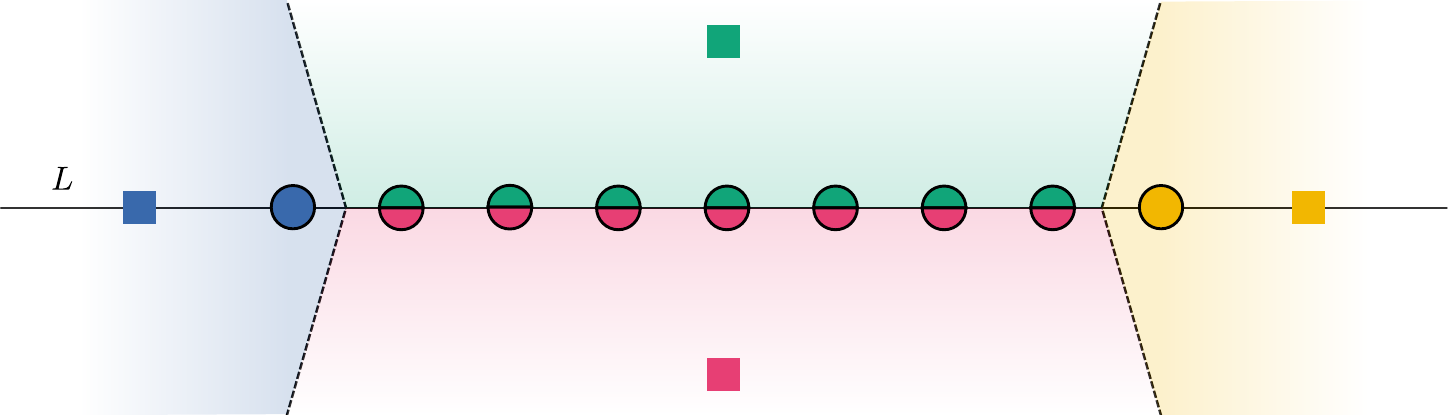}
    \caption{A clustering and diagram that are strongly compatible. Sites are depicted as squares. The clustering assigns all points on the boundary between the green and pink cluster with a proportion of $\nicefrac{1}{2} : \nicefrac{1}{2}$ to the two clusters. Consequently, a dissection of the horizontal line $L$ with the required properties of \cref{th:generalized_structure-generic-clustering} requires more than $2k-1 = 2\cdot4-1 = 7$ intervals.}
    \label{fig:s_compatible_does_not_suffice}
\end{figure}
In this case, the number of splitting operations will generally not be bounded by a function of $k$ only, and \cref{th:generalized_structure-generic-clustering} does not hold anymore.  

In fact, a clustering that admits a feasible or even strongly compatible diagram does not suffice to provide the bound required later, and we need the additional properties of strictly compatible diagrams to prove the following theorem. 

\begin{theorem} \label{th:generalized_structure}
	Let $\bar{C} \in \Cs_{\K}(k,\bar{X},\Omega)$ admit a strictly compatible diagram, and let $L$ be a line in $\Rbb^d$. Then, there is a partition of $L$ into at most $2k-1$ proper intervals such that all data points in the relative interior of each interval are fully assigned to the same cluster.
\end{theorem}

\begin{proof}
Let $\apd=\apd(T,\Sigma,\mathcal{A})=\{P_1,\ldots,P_k\}$ be a diagram such that $\bar{C}$ and $\apd$ are strictly compatible. Suppose that exactly $k'$ of the univariate polynomials $g_i|_L$ are different, let $N^+\subset [k]$ be the set of the corresponding indices, and set $N^-=[k]\setminus N^+$. 

If $k'=k$, the assertion follows from \cref{th:generalized_structure-generic-clustering}. So, suppose that $k' < k$. By \cref{th:generalized_structure-generic-clustering}, we obtain a partition $\mathcal{I}$ of $L$ into at most $2k'-1$ proper intervals with the following property: If $I \in \mathcal{I}$ and $\tilde{x} \in \tilde{X} \cap \relint(I)$, then $\tilde{x}\not\in P_i$ for any $i\in N^+$.

Since $\bar{C}$ and $\apd$ are strictly compatible, $P_i \cap P_\ell$ contains at most one point for $i\ne \ell$. Since, for each $i\in N^-$, each $g_i|L$ coincides with one of the $g_\ell|L$ with $\ell \in N^+$, splitting intervals for each such points will result in an additional number of at most $|N^-|$ intervals in order to guarantee the desired property.
Hence we obtain a total of at most
    \begin{equation*}
    2 |N^+|-1 + |N^-| = 2k'-1 + k- k' = k+ k'-1 \le 2k-1
    \end{equation*}
intervals that have the required properties.
\end{proof}

\subsection{Proof of the coreset properties} \label{ssec:coreset-properties}
We will now prove various lemmas that will subsequently be used to show that $(\coresetX,\coresetOmega)=(\batchingX,\batchingOmega)$ is a coreset for $(k,X,\Omega, \Acc,\Kappa)$, thereby proving \cref{th:smaller_coreset}. Note that, by \cref{le:projecting_is_a_coreset} and \cref{le:coresets_give_coresets}, it suffices to show that $(\batchingX,\batchingOmega)$ is a coreset for $(k,\bar{X},\Omega, \Acc,\Kappa)$.

We begin with a simple observation that relates the batch errors in the anisotropic case to those in the Euclidean case.

\begin{remark}\label{le:euclidean-anisotrop} 
Let $\batches'\subset \batches$. Then,
\begin{equation*}
    \lambdamin \cdot \sum_{B \in \batches'} \error_{E}(B) \le 
    \sum_{i=1}^{k} \smashoperator[r]{\sum_{B \in \batches'} } \xi_{iB}  \error_{A_i}(B)
    \le \lambdamax \cdot \sum_{B \in \batches'} \error_{E}(B).
\end{equation*}
\end{remark}

\begin{proof}
By \cref{rem:norm_equivalence} and since $\sum_{i=1}^{k} \xi_{iB}=1$ for each $B \in \batches'$, we obtain 
	\begin{equation*}
	    \sum_{i=1}^{k}\sum_{B \in \batches'} \xi_{iB} \error_{A_i}(B)
	    \le \lambdamax \sum_{i=1}^{k}\sum_{B \in \batches'} \xi_{iB}  \error_{E}(B)
	   = \lambdamax \sum_{B \in \batches'} \error_{E}(B)
	\end{equation*}
and	
	\begin{equation*}
 \sum_{i=1}^{k} \smashoperator[r]{\sum_{B \in \batches'} } \xi_{iB}  \error_{A_i}(B) \ge \lambdamin \sum_{i=1}^{k} \smashoperator[r]{\sum_{B \in \batches'} } \xi_{iB}  \error_{E}(B) = \lambdamin \smashoperator[r]{\sum_{B \in \batches'} }  \error_{E}(B).
\end{equation*}
\end{proof}

The following lemma essentially handles condition \cref{def:property_a} of \cref{def:coreset} for each set $S$ of sites.

\begin{lemma} \label{le:proof_property_a} 
Let $\batchingC \in \Cs_\Kappa(k,\batchingX,\batchingOmega)$. Then,
    \begin{align*}
	\cost_\AK(\bar{X},\extension_2(\batchingC),S) \le \cost_\AK(\batchingX,\batchingC,S) + \lambdamax \cdot \sum_{B \in \batches} \error_{E}(B).
	\end{align*}
\end{lemma}

\begin{proof}
First note that, for each $B\in \batches$ we have by \cref{le:center_replacement} \ref{le:assertion_a} (with $X_0=B$, $c_0=c_B$)
\begin{equation*}
  \omega_{B} \norm{x_{B} - s_i}_{A_i}^{2}=  \sum_{\bar{x}_{j} \in B} \omega_{j} \norm{\bar{x}_{j}-s_i}_{A_i}^{2} - \error_{A_i}(B).
\end{equation*}
Therefore,
	\begin{align*}
	\cost_\AK(\batchingX,\batchingC,S) & = \sum_{i=1}^{k} \sum_{B\in \batches} \xi_{iB} \omega_{B} \norm{x_{B} - s_i}_{A_i}^{2}\\
	 & = \sum_{i=1}^{k} \sum_{B \in \batches}  \Biggl( \sum_{\bar{x}_{j} \in B} \xi_{iB} \omega_{j} \norm{\bar{x}_{j}-s_i}_{A_i}^{2} - \xi_{iB}\error_{A_i}(B) \Biggr).
	\end{align*}
Since $\bar{\xi}_{ij}=\xi_{iB}$ for each $\bar{x}_j\in B$ we have
	\begin{align*}
	\sum_{i=1}^{k} \sum_{B \in \batches} \sum_{\bar{x}_{j} \in B} & \xi_{iB} \omega_{j} \norm{\bar{x}_{j}-s_i}_{A_i}^{2} 
	 = \sum_{i=1}^{k} \sum_{B \in \batches}  \sum_{\bar{x}_{j} \in B} \bar{\xi}_{ij} \omega_{j} \norm{x_{j}-s_i}_{A_i}^{2} \\
	& =\sum_{i=1}^{k} \sum_{j=1}^{n}   \bar{\xi}_{ij} \omega_{j} \norm{\bar{x}_{j}-s_i}_{A_i}^{2} = \cost_\AK(X,\extension_2(\batchingC),S).
	\end{align*}
Hence,
	\begin{align*}
	\cost_\AK(\batchingX,\batchingC,S) = \cost_\AK(X,\extension_2(\batchingC),S) - \sum_{i=1}^{k}\sum_{B \in \batches} \xi_{iB} \error_{A_i}(B),
	\end{align*}
and, by \cref{le:euclidean-anisotrop}, this yields
	\begin{equation*} 
	\cost_\AK(X,\extension_2(\batchingC),S) \le \cost_\AK(\batchingX,\batchingC,S) + \lambdamax \sum_{B \in \batches}  \error_{E}(B),
	\end{equation*}
which is the asserted inequality.
\end{proof}

In order to verify that $(\batchingX,\batchingOmega)$ is a coreset for $(k,\bar{X},\Omega, \Acc,\Kappa)$, we address condition \cref{def:property_b} of \cref{def:coreset}. With the aid of \cref{le:restriction_to_diagrams}, let $ \bar{C}^* \in \Cs_{\K}(k, \bar{X}, \Omega) $ such that it attains $\cost_\AK(\bar{X},S)$ and admits a strictly compatible anisotropic diagram $\apd=\apd(T,\Sigma,\mathcal{A})$.  
We will show that, for $\batchingC = p(\bar{C}^*)$,  
\begin{equation}\label{eq:to_show_for_b}
    \cost_\AK(X_\batches,\batchingC,S) + \Delta^{-} \le (1+\epsilon) \cost_\AK(\bar{X},\bar{C}^*,S).
\end{equation}
as this yields the inequality
\begin{align*}
\cost_\AK(X_\batches,S)+ \Delta^{-} &\le     \cost_\AK(X_\batches,\batchingC,S) + \Delta^{-} \\
& \le (1+\epsilon) \cost_\AK(\bar{X},\bar{C}^*,S)
= (1+\epsilon) \cost_\AK(\bar{X},S),
\end{align*}
needed for condition \cref{def:property_b}.
The next two lemmas will prepare the ground to verify the former inequality.

The clustering $\bar{C}^*$ partitions $\batches$ into the set $\batches^{+}$ of those batches $B$ whose points are all assigned to the same cluster, and its complement $\batches^{-}=\batches \setminus \batches^{+}$. Similarly, $(\bar{X},\Omega)$ and $(X_\batches,\Omega_\batches)$ are split into $(\bar{X}|_\pm,\Omega|_\pm)$ and $(X_\batches|_\pm,\Omega_\batches|_\pm)$, and the clusterings $\bar{C}^*$ and $\batchingC$ are split into their restrictions $\bar{C}^*|_\pm$ and $\batchingC |_\pm$ of $(\bar{X}|_\pm,\Omega|_\pm)$ and $(X_\batches|_\pm,\Omega_\batches|_\pm)$, respectively. As the clustering cost is additive in the contributions of each point, we can consider the costs
\begin{align*}
\cost_\AK(\bar{X}|_{\pm},\bar{C}^*|_{\pm},S) &= \sum_{i=1}^{k} \sum_{B \in \batches^{\pm}} \smashoperator[r]{\sum_{\bar{x}_{j}\in B}} \bar{\xi}_{ij} \omega_{j} \norm{\bar{x}_{j}-s_{i}}_{A_i}^{2},\\
\cost_\AK(\batchingX|_\pm,\batchingC |_\pm,S) &= \sum_{i=1}^{k} \smashoperator[r]{\sum_{B \in \batches^{\pm}}}  \xi_{iB} \omega_{B} \norm{x_{B}-s_{i}}_{A_i}^{2}
\end{align*}
incurred by points in batches from $\batches^{+}$ and $\batches^{-}$ separately. (Let us point out that the subscript $\Kappa$ is used here to indicate that we are still in the constrained case. The bounds in $\Kappa$, however, only apply to the clusterings on the unsplit data sets.)

Note that, for each $B\in \batches^{+}$,
\begin{equation*}
    \xi_{iB}= \bar{\xi}_{ij} = 
    \begin{cases}
    1 & \text{for $\bar{x}_j\in B$};\\
    0 & \text{for $\bar{x}_j\not\in B$}.
    \end{cases}
\end{equation*}
This property is used in the next lemma to show that batches in $\batches^{+}$ behave nicely. 

\begin{lemma} \label{le:proof_property_b_good}
\begin{equation*}
    \cost_\AK(\batchingX|_+,C_\batches |_+,S) + \lambdamin \smashoperator[r]{\sum_{B \in \batches^{+}} } \error_{E}(B) \le \cost_\AK(\bar{X}|_{+},\bar{C}^*|_{+}, S).
\end{equation*}
\end{lemma}

\begin{proof}
    Since all points of each batch $B \in \batches^{+}$ have been assigned to the same cluster, we see with the aid of \cref{le:center_replacement} \ref{le:assertion_a} (with $X_0=B$, $c_0=x_B$) that
\begin{align*}
	\cost_\AK &(X|_{+},\bar{C}^*|_{+}, S) = \sum_{i=1}^{k} \sum_{B \in \batches^{+}} \smashoperator[r]{\sum_{\bar{x}_{j}\in B}} \bar{\xi}_{ij} \omega_{j} \norm{\bar{x}_{j}-s_i}_{A_i}^{2}\\
	&= \sum_{i=1}^{k} \smashoperator[r]{\sum_{B \in \batches^{+}}} \xi_{iB}  \smashoperator[r]{\sum_{\bar{x}_{j}\in B}} \omega_{j} \norm{\bar{x}_{j}-s_i}_{A_i}^{2} 
	= \sum_{i=1}^{k} \smashoperator[r]{\sum_{B \in \batches^{+}}} \xi_{iB} \left( \omega_{B} \norm{x_{B}-s_i}_{A_i}^{2} + \error_{A_i}(B) \right)\\
	&= \cost_\AK(\batchingX|_+,\batchingC |_+,S) + \sum_{i=1}^{k} \smashoperator[r]{\sum_{B \in \batches^{+}} } \xi_{iB}  \error_{A_i}(B).  
	\end{align*} 
Now, the asserted inequality follows with the aid of \cref{le:euclidean-anisotrop}.
\end{proof}

While the proof of \cref{le:proof_property_b_good} did not make any use of the special choice of the optimal clustering, the next lemma, which deals with $\batches^{-}$, relies heavily on the property that $\bar{C}^*$ admits a strictly compatible diagram. (Note that \cref{le:proof_property_b_bad} can also be used to fix a flaw in the proof \cite{Har-Peled2007}[Theorem 3.7].)

\begin{lemma}  \label{le:proof_property_b_bad}
\begin{align*}
     \cost_\AK(\batchingX|_-,& \batchingC|_-,  S)  +  	\lambdamin \sum_{B \in \batches^{-}} \error_{E}(B) \\ 
     &\le \cost_\AK(\bar{X}|_-,\bar{C}^*|_-,S) +  \epsilon \cost_\AK(\bar{X},\bar{C}^*,S).
\end{align*}
\end{lemma}

\begin{proof}
	First, by the definition of the merging function $p$,
	\begin{equation*}
	    \xi_{iB} \omega_{B}= \sum_{\bar{x}_j\in B} \bar{\xi}_{ij} \omega_{j}.
	\end{equation*}
Hence,
\begin{equation*}
	    \cost_\AK(\batchingX|_-, \batchingC|_-,  S)  
	    = \sum_{i=1}^k \sum_{B\in \batches^-} \sum_{\bar{x}_j\in B} \bar{\xi}_{ij} \omega_{j} \norm{x_{B} -s_{i}}_{A_i}^{2}.
\end{equation*}
Next, we apply \cref{le:technical_bound} to $(\bar{X}|_-,\Omega|_-)$ and $\bar{C}^*|_-$. In this setting, the number $D$ becomes
\begin{equation*}
    D = \sum_{i=1}^k \sum_{B\in \batches^-} \sum_{\bar{x}_j\in B} \bar{\xi}_{ij} \omega_{j} \norm{x_{B} -\bar{x}_j}_{A_i}^{2},  
\end{equation*}
and \cref{le:technical_bound} yields
\begin{equation*}
	    \cost_\AK (\batchingX|_-, \batchingC|_-,  S)
	    \le \cost_\AK(\bar{X}|_-,\bar{C}^*|_-,S) + 2 
	    \sqrt{D \cdot \cost_\AK(\bar{X}|_-,\bar{C}^*|_-,S)} + D.
\end{equation*}
We will now apply \cref{le:euclidean-anisotropic-bound} to derive an upper bound for $D$. In order to do so, let us first consider the lemma's assumption. In fact, we have
\begin{equation*}
	\sum_{B \in \batches^-} \sum_{x_j \in B}  \omega_{j} \norm{x_{B}-\bar{x}_{j}}_{2}^{2} = \sum_{B \in \batches^{-}} \error_{\identity}(B) \le |\batches^-| \cdot \error_{\identity}= \frac{|\batches^-| \cdot \epsilon^2}{32\alpha k \lambdamax |\mathcal{L}|} \text{ALG}_{(\alpha,\beta)}(X).
\end{equation*}
Hence we need to bound the number of batches in $\batches^{-}$. We consider the different lines of the pencils separately.

So, let $L\in \mathcal{L}$. Since $\bar{C}^*$ admits a strictly compatible diagram we can apply \cref{th:generalized_structure}. There is thus a partition of $L$ into at most $2k-1$ proper intervals such that all data points in the relative interior of each interval are fully assigned to the same cluster. Consequently, only batches that contain a boundary point of at least one such interval can contain points of more than one cluster. Therefore, $|\batches^-\cap \batches(L)| \le 2(k-1)$, and thus
\begin{equation*}
    |\batches^-| \le 2(k-1) \cdot |\mathcal{L}|.
\end{equation*}
Since $(\bar{X},\Omega)$ is a linear $\epsilon$-coreset for $(X,\Omega)$ and $\epsilon \in (0,\nicefrac{1}{2}]$, we obtain
\begin{align*}
	\sum_{B \in \batches^-} \sum_{x_j \in B}  \omega_{j} \norm{x_{B}-\bar{x}_{j}}_{2}^{2}
	&\le \frac{\epsilon^2}{16\lambdamax} \OPT(X)
	\le \frac{\epsilon^2}{8\lambdamax} \OPT(\bar{X}).
\end{align*}
\cref{le:euclidean-anisotropic-bound}, applied to $(\bar{X}, \Omega)$ with $\eta = \nicefrac{\epsilon^2}{8}$ yields 
	\begin{equation*}
	   D= \sum_{i=1}^k \sum_{B\in \batches^-} \sum_{\bar{x}_j\in B} \bar{\xi}_{ij} \omega_j \|x_{B}- \bar{x}_j\|^2_{A_i} \le \frac{\epsilon^2}{8} \OPT_\Acc(\bar{X}) \le \frac{\epsilon^2}{8} \cost_\AK(\bar{X},\bar{C}^*,S).
	\end{equation*}
Since
\begin{equation*}
    \cost_\AK(\bar{X}|_-,\bar{C}^*|_-,S) \le \cost_\AK(\bar{X},\bar{C}^*,S),
\end{equation*}
we obtain
\begin{align*}
	    \cost_\AK (\batchingX|_-, \batchingC|_-,  S)
	    & \le \cost_\AK(\bar{X}|_-,\bar{C}^*|_-,S) + \left(2 
	    \sqrt{\frac{\epsilon^2}{8}} + \frac{\epsilon^2}{8}\right) \cost_\AK(\bar{X},\bar{C}^*,S).
\end{align*}
Finally, with the aid of \cref{le:euclidean-anisotrop}, the previous estimates also imply that
\begin{align*}
	\sum_{i=1}^{k} \sum_{B \in \batches^{-}} \xi_{iB}  \error_{A_i}(B) & \le \lambdamax\cdot\sum_{B \in \batches^{-}}  \error_{E}(B) 
	= \lambdamax\cdot \sum_{B \in \batches^-} \sum_{x_j \in B}  \omega_{j} \norm{x_{B}-\bar{x}_{j}}_{2}^{2} \\
	&\le \frac{\epsilon^2}{8} \cost_\AK(\bar{X},\bar{C}^*,S).
\end{align*}
Therefore, using \cref{le:euclidean-anisotrop} again,
\begin{align*}
 \cost_\AK(\batchingX|_-, &\batchingC|_-,  S)  +  \lambdamin	\sum_{B \in \batches^{-}} \error_{E}(B)\\ 
 &\le 
	   \cost_\AK(\batchingX|_-, \batchingC|_-,  S)  +  	\sum_{i=1}^{k} \smashoperator[r]{\sum_{B \in \batches^{-}} } \xi_{iB}  \error_{A_i}(B) \\
	   & \le \cost_\AK(\bar{X}|_-,\bar{C}^*|_-,S) + \left(2 
	    \sqrt{\frac{\epsilon^2}{8}} + 2\cdot \frac{\epsilon^2}{8}\right) \cost_\AK(\bar{X},\bar{C}^*,S)\\
	   &\le \cost_\AK(\bar{X}|_-,\bar{C}^*|_-,S) + \epsilon \cdot \cost_\AK(\bar{X},\bar{C}^*,S),
\end{align*} 
as claimed.
\end{proof}

Now, we are ready to verify that $(\batchingX,\batchingOmega)$ is a coreset for $(k,\bar{X},\Omega, \Acc,\Kappa)$.

\begin{theorem} \label{th:batching_makes_a_coreset}
$(\batchingX,\Omega_{\batches})$ is an $\left(\epsilon,\frac{\lambdamax}{\lambdamin}\right)$-coreset for $(k,\bar{X},\Omega, \Acc,\Kappa)$.
\end{theorem}

\begin{proof}
Let 
\begin{equation*}
    \Delta^{+} = \lambdamax \cdot \sum_{B \in \batches}   \error_{E}(B), \quad \text{and} \quad
    \Delta^{-} = \lambdamin \cdot \sum_{B \in \batches}   \error_{E}(B).
\end{equation*}
Then,
\begin{equation*}
    \Delta^{+} =  \frac{\lambdamax}{\lambdamin} \cdot \Delta^{-}.
\end{equation*}
Further, \cref{le:proof_property_a} immediately yields
	\begin{equation*} 
	\cost_\AK(\bar{X},f_2(\batchingC),S) \le \cost_\AK(\batchingX,\batchingC,S) + \Delta^{+},
	\end{equation*}
which shows condition \cref{def:property_a} of \cref{def:coreset}. 

In order to verify \cref{def:property_b}, remember that it suffices to show:

\begin{equation*}
    \cost_\AK(X_\batches,\batchingC,S) + \Delta^{-} \le (1+\epsilon) \cost_\AK(\bar{X},\bar{C}^*,S).
\end{equation*}

We combine \cref{le:proof_property_b_good,le:proof_property_b_bad}, using  
\begin{align*}
	\cost_\AK(\batchingX,\batchingC, S) + \Delta^{-}= &\cost_\AK(X_{\batches|+},C_{\batches|+}, S)  + \lambdamin \smashoperator[r]{\sum_{B \in \batches^{+}}} \error_{E}(B) \\ 
	& + \cost_\AK(\batchingX|_-, \batchingC|_-, S) + \lambdamin \smashoperator[r]{\sum_{B \in \batches^{-}}} \error_{E}(B).
\end{align*}
This yields
\begin{align*}
    &\cost_\AK(\batchingX,\batchingC, S) + \Delta^{-} \\
    &\quad\le \cost_\AK(\bar{X}|_{+},\bar{C}^*|_{+}, S) +  \cost_\AK(\bar{X}|_-,\bar{C}^*|_-,S) +  \epsilon \cost_\AK(\bar{X},\bar{C}^*,S) \\
    &\quad\le (1+\epsilon)\cost_\AK(\bar{X},\bar{C}^*, S),
\end{align*}
which implies condition \cref{def:property_b} of \cref{def:coreset}.
\end{proof}

Finally, we have all the ingredients with which to prove our main theorem.

\begin{proof}[Proof of \cref{th:smaller_coreset}]
The assertion that $(\coresetX,\coresetOmega)=(\batchingX,\batchingOmega)$ is a coreset for the instance $(k,X,\Omega, \Acc,\Kappa)$ follows immediately from our previous results by means of 
\cref{le:coresets_give_coresets}.

More precisely, we apply \cref{le:projecting_is_a_coreset} for $\epsilon_1=\nicefrac{\epsilon}{3}$ to obtain a linear coreset,i.e., $\delta_1=1$, and \cref{th:batching_makes_a_coreset}  for $\epsilon_2=\nicefrac{\epsilon}{3}$ and $\delta_2=\nicefrac{\lambdamax}{\lambdamin}$. Since
\begin{equation*}
    \epsilon \ge \epsilon_1+\epsilon_2+\epsilon_1\epsilon_2 \quad \text{and} \quad \delta=\max \{\delta_1,\delta_2\} = \frac{\lambdamax}{\lambdamin},
\end{equation*} 
the assertion follows from \cref{le:coresets_give_coresets}.
\end{proof} 

\section{Large Sensitivity} \label{sec:counterexample}
While \cref{th:smaller_coreset} provides  what is currently the best-known bound for the size of deterministic coresets for \wcac, we will round off the paper by addressing the question of whether it is possible to design even smaller coresets for \wcac by applying techniques based on the notion of sensitivity within a probabilistic setting.

In \emph{importance sampling}, as applied to unconstrained least-squares clustering, points are sampled from the data set according to their relevance for the clustering cost. The importance $t(x_j)$ of a point $x_j\in X$, called the \emph{sensitivity} of $x_j$, is assessed as its contribution, in the worst case, to the cost of any feasible clustering. More formally, for the unconstrained case, 
\begin{equation*}
t(x_j) = \sup_{S} \,\, \frac{1}{\cost(X,S)} \sum_{i=1}^k \xi_{ij} \omega_j \norm{x_j-s_i}_2^2,
\end{equation*}
where $C(S)=(\xi_{11},\ldots,\xi_{kn}) \in \argmin \{\cost(X,C,S): C \in \Cs(k,X,\Omega)\}$, and the supremum is taken over all sets $S$ of $k$ sites. Moreover,
\begin{equation*}
    T= T(k,X,\Omega) = \sum_{j=1}^{n} t(x_{j})
\end{equation*}
is called \emph{total sensitivity}. Note that $t(x_j) \le 1$, hence $T\le n$. A key requirement for successfully employing the concept of importance sampling for obtaining small coresets is, however, that $T$ does not depend on the cardinality $n$ of the data set but only on $k$, see \cite[Lemma 2.2]{Bachem2017a}. In fact, \cite[Theorem 3.1]{Langberg2010} showed that $T \in \mathcal{O}(k)$ and derived coresets of size nearly quadratic in $k$ with high probability. Subsequently, \cite[Theorem 6.7]{Braverman2016}  used an improved bound to obtain coresets of size nearly linear in $k$ via importance sampling. To our knowledge, this is the smallest dependence on $k$ currently known.

While such techniques have been successfully applied to unconstrained least-squares clustering, it is not clear, and it is indeed posed as an open question in \cite{Huang2019a, Schmidt2020}, whether they still work in the presence of constraints for the cluster weights. 
 
Of course, the notion of sensitivity can easily be generalized to the constrained case, yielding $t_\Kappa(x_j)$, by simply replacing $\cost(X,S)$ by $\cost_\Kappa(X,S)$ and taking the $\argmin$ of $\cost_\Kappa(X,C, S)$ for all $C \in \Cs_\Kappa(k,X,\Omega)$. As the following \cref{ex:sensitivity} shows, the total sensitivity $T_\Kappa$ might, however be as large as $n$. Hence approaches based on sampling require too large samples to be of any use for \wcac.

\begin{example}\label{ex:sensitivity}
For $r\in (0,\nicefrac{1}{2})$ let $x_1,\ldots,x_n\in \mathbb{R}^{2}$ be equally spaced on the circle with radius $r < 1$ centered at the origin, and set
\begin{equation*}
    X=\{x_1,\ldots,x_n\}\subset \mathbb{R}^{2}, \quad \text{and} \quad \Omega= \{1,\ldots,1\}.
\end{equation*}
Further, let 
\begin{equation*}
    k=2, \quad \kappa_{1}^-=\kappa_{1}^+=n-1, \quad \kappa_{2}^-=\kappa_{2}^+=1, \quad \text{and} \quad \Kappa=\{n-1,n-1,1,1\}.
\end{equation*}
Now, we choose a point $x_{j_0} \in X$ and set
\begin{equation*}
    s_{1}=0, \quad s_{2}=\frac{1}{r} x_{j_0}  \quad \text{and} \quad S_0=\{s_1,s_2\};
\end{equation*}
see \cref{fig:sensitivitycounterexample}. We are interested in finding a clustering $C=(C_1,C_2)\in \Cs_\Kappa(2,X,\Omega)$ that minimizes $\cost_\Kappa(X,C,S_0)$. 

\begin{figure}[h]
	\centering
	\includegraphics[width=0.7\linewidth]{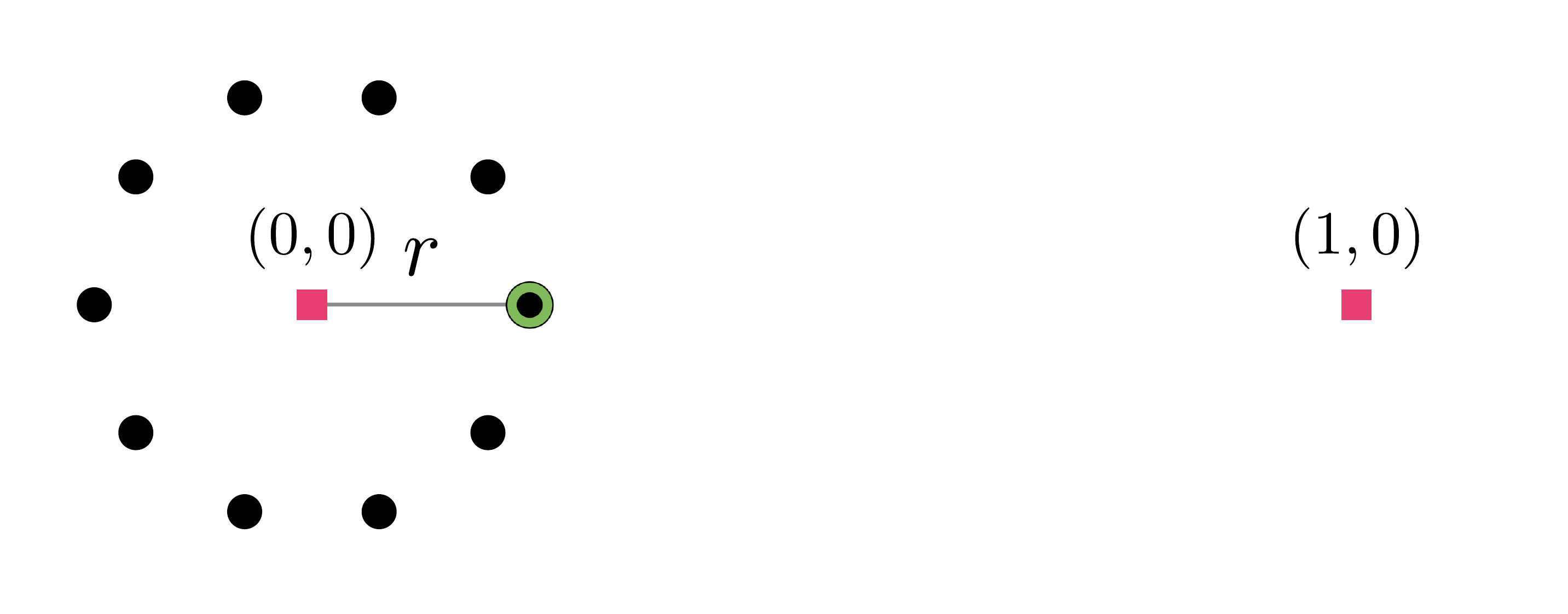}
	\caption{Construction of \cref{ex:sensitivity}. The data points are represented as black circles and the two sites $s_1=0$ and $s_2=(1,0)^\top$ are indicated as pink squares. The encircled data point $x_{j_0}$ contributes most of the cost.}
	\label{fig:sensitivitycounterexample}
\end{figure}

Since $x_{j_0}$ is the unique closest point of $x_{j_0}$ to $s_2$ and $C_2$ requires exactly weight one, the optimal clustering is unique, each point is fully assigned to one cluster, and $C_1$ contains the points of $X\setminus \{x_{j_0}\}$ while $C_2$ consists of $x_{j_0}$. More formally,
\begin{equation*}
    \xi_{ij} = 
    \begin{cases}
    1 & \text{for $(i,j)\in\bigl\{(1,j): j\in [n]\setminus \{j_0\}\bigr\}\cup  \{(2,j_0)\}$};\\
    0 & \text{else},
    \end{cases}
\end{equation*}
and we obtain
\begin{equation*}
    \cost_\Kappa(X,S_0)= \sum_{i=1}^2 \sum_{i=1}^n \xi_{ij} \norm{x_j-s_i}_2^2 = (n-1)r^2 + (1-r)^2.
\end{equation*}
Hence, for the sensitivity $t_\Kappa(x_{j_0})$ of $x_{j_0}$, we know
\begin{equation*}
t_\Kappa(x_{j_0}) = \sup_{S} \,\, \frac{1}{\cost_\Kappa(X,S)} \sum_{i=1}^2 \xi_{ij_0} \norm{x_{j_0}-s_i}_2^2
\ge \frac{(1-r)^2}{(n-1)r^2 + (1-r)^2}.
\end{equation*}
Since
\begin{equation*} 
\lim_{r \to 0}t_\Kappa(x_{j_0}) \ge \lim_{r \to 0} \frac{(1-r)^2}{(n-1)r^2 +(1-r)^2} =  1
\end{equation*}
and $x_{j_0}\in X$ was chosen arbitrarily, we see that
\begin{equation*}
   n\ge T_\Kappa= T_\Kappa(k,X,\Omega) = \sum_{j=1}^{n} t_\Kappa (x_{j}) \ge n,
\end{equation*}
proving \cref{th:unbounded_sensitivity}.
\end{example}

Note that in the instance of \wcac of \cref{ex:sensitivity}, the weight constraints for the clusters prevent the assignment of each point to its nearest site. In fact, for sufficiently small $r$, the one point assigned to the \enquote{far-away site} $s_2$ contributes the most to the clustering cost. Since the sensitivity measures the worst-case contribution over all choices of sites, each of the $n$ points can be decisive for the cost. Consequently, importance sampling, based on the above notion of sensitivity, does not help in designing smaller coresets for \wcac.

\section{Final remarks} \label{conclusion}

Our results provide small coresets for \wcac, i.e., in the weight-constrained anisotropic case. This allows to compute good clusterings for the generally much smaller sets $(\coresetX,\coresetOmega)$ and to subsequently convert them to clusterings of the original data sets $(X,\Omega)$. While the corresponding extensions $\extension$ preserve the cluster weights and guarantee the feasibility of the obtained clusterings, we might, however, loose the favorable properties of compatibility with diagrams. Hence, the question arises as to whether any efficiently computable extensions exist that map clusterings that admit a (strictly, strongly, or just) compatible diagram on the coresets to clusterings of the same type on the full data set. Alternatively, we might ask if we can find a diagram on the coreset such that the induced clustering on the original data set only slightly violates the cluster size constraints. This leads to a number of relevant stability issues.

One may also wonder whether other techniques might enable the design of even smaller coresets for \wcac. While \cref{th:unbounded_sensitivity} gives a negative answer for techniques based on importance sampling, even for weight-constrained least-squares clustering, there might be other techniques that are better suited for constrained clustering.
Recall that the high sensitivity of each data point in  \cref{ex:sensitivity} is a consequence of its definition as worst-case behavior. Of course, the average or expected contribution of each point to the clustering cost taken over all choices of $k$ sites is much smaller in this example. Therefore, it might be reasonable to investigate whether such a weaker notion of average-case sensitivity can be utilized to design improved coresets.

For the applications in materials science referred to in the introduction, it might be worthwhile to point out that our construction and analysis in the proof of \cref{th:smaller_coreset}, which has now reduced the dependency on $k$ from cubic in \cite{Har-Peled2007} to quadratic, does not come with a larger constant hidden in the big-O notation. In fact, by representing each batch of points by its centroid, the number of coreset points is actually reduced by half (note that  \cite{Har-Peled2007} states explicitly that this does not work in their setup). Reductions in the multiplicative constant play a decisive role in the grain map application in practice in terms of memory requirements and computation times. We refer to \cite{AFGK21a} for a computational study of grain map reconstructions based on coresets.

Finally, let us point out that our approach is not limited to \wcac. In fact, it will work as long as the extensions $\extension$ preserve the feasibility for the constraint set. This is the case, for instance, with \emph{fair clustering}, where points belong to different groups and clusters should not over- or under -represent any of the groups. Using a result by \cite{Huang2019a} that designing coresets for this problem with $l$ groups can be reduced to the design of $l$ coresets for points from a single group, our results can be extended by creating batches per group.  

\printbibliography

\appendix

\end{document}

%% file: packages.tex
\usepackage{centernot}
\usepackage{booktabs}
\usepackage{multirow,makecell,tabularx,tabu}

\usepackage{enumitem}
\setlist[enumerate]{nolistsep}
\usepackage{graphicx}
\usepackage{geometry}
\usepackage{algorithm2e}
\usepackage{tikz,pgfplots}
\usetikzlibrary{calc}
\usetikzlibrary{intersections}
\usepackage{mathtools}
\usepackage{xcolor}
\usepackage[maxnames=8,backend=biber,isbn=false,url=false,eprint=false]{biblatex}
\usepackage{caption}
\usepackage{subcaption}
\usepackage{amsfonts,amsthm,latexsym,amsmath,amssymb,amscd,amsmath,epsf,mathtools,nccmath,bm}
\usepackage[capitalise]{cleveref}
\usepackage[draft,margin]{fixme}
\fxsetup{theme=color}
\usepackage[utf8]{inputenc}
\usepackage{setspace}
\usepackage{xspace}
\usepackage{nicefrac}
\usepackage{csquotes}

%% file: macros.tex
\newtheorem{theorem}{Theorem}[section]
\newtheorem*{theorem*}{Theorem}
\newtheorem{proposition}[theorem]{Proposition}
\newtheorem{corollary}[theorem]{Corollary}
\newtheorem*{corollary*}{Corollary}
\newtheorem{lemma}[theorem]{Lemma}

\theoremstyle{definition}

\newtheorem{remark}[theorem]{Remark}
\newtheorem{definition}[theorem]{Definition}
\newtheorem{example}[theorem]{Example}
\crefname{equation}{}{}

\FXRegisterAuthor{pg}{apg}{PG}
\FXRegisterAuthor{mf}{amf}{MF}

\DeclarePairedDelimiter\norm{\lVert}{\rVert}
\DeclareMathOperator{\cost}{cost}

\newcommand{\OPT}{\operatorname{OPT}}
\newcommand{\ALG}{\operatorname{ALG}}
\DeclareMathOperator{\K}{\mathrm{K}}

\DeclareMathOperator{\relint}{relint}

\newcommand{\Rbb}{\mathbb{R}}
\newcommand{\Nbb}{\mathbb{N}}
\newcommand{\R}{\mathbb{R}}
\newcommand{\N}{\mathbb{N}}

\newcommand{\Acc}{\mathcal{A}}
\newcommand{\Ecc}{\mathcal{E}}

\newcommand{\Cs}{\mathcal{C}}

\newcommand{\support}{\mathrm{supp}}
\newcommand{\Kappa}{\mathrm{K}}

\newcommand{\lambdamax}{\lambda^{+}(\Acc)}
\newcommand{\lambdamin}{\lambda^{-}(\Acc)}
\newcommand{\error}{\mathcal{V}}
\newcommand{\AK}{{\Acc,\Kappa}}

\newcommand{\apd}{\mathcal{P}}

\newcommand{\coresetX}{\tilde{X}}
\newcommand{\coresetx}{\tilde{x}}

\newcommand{\coresetOmega}{\tilde{\Omega}}

\newcommand{\coresetomega}{\tilde{\omega}}
\newcommand{\coresetC}{\tilde{C}}

\newcommand{\extension}{f}
\newcommand{\batchingX}{X_{\mathcal{B}}}
\newcommand{\batchingOmega}{\Omega_{\mathcal{B}}}
\newcommand{\batchingC}{C_{\mathcal{B}}}
\newcommand{\batches}{\mathcal{B}}

\newcommand{\wcaa}{wca assignment\xspace}
\newcommand{\wcac}{wca clustering\xspace}
\newcommand{\wcaas}{wca assignments\xspace}
\newcommand{\wcacs}{wca clusterings\xspace}
\newcommand{\instance}{I}
\newcommand{\coresetinstance}{\tilde{\instance}}

\newcommand{\identity}{E}

\DeclareMathOperator*{\argmax}{argmax}
\DeclareMathOperator*{\argmin}{argmin}

\makeatletter
\newcommand{\leqnomode}{\tagsleft@true}
\makeatother

\definecolor{lBlue}{HTML}{a6cee3}
\definecolor{dBlue}{HTML}{1f78b4}
\definecolor{lGreen}{HTML}{b2df8a}
\definecolor{dGreen}{HTML}{33a02c}
\definecolor{lRed}{HTML}{fb9a99}
\definecolor{dRed}{HTML}{e31a1c}
\definecolor{lOrange}{HTML}{fdbf6f}
\definecolor{dOrange}{HTML}{ff7f00}
\definecolor{lPurple}{HTML}{cab2d6}
\definecolor{dPurple}{HTML}{6a3d9a}

\pgfplotsset{compat=1.15}
\pgfplotsset{data/.style={xtick distance=1,ytick distance=1,ticks=none,axis line style={draw=none},grid=major,axis equal image}}